\documentclass[11pt]{article}
\usepackage{amsmath, amssymb}
\usepackage{amsfonts}
\usepackage[utf8]{inputenc}
\usepackage{amsthm}
\usepackage{graphicx}
\usepackage{csquotes}
\usepackage{xspace}
\usepackage{fullpage}

\graphicspath{{picss/}}

\newtheorem{theorem}{Theorem}[section]
\newtheorem{definition}[theorem]{Definition}
\newtheorem{lemma}[theorem]{Lemma}
\newtheorem{corollary}[theorem]{Corollary}
\newtheorem{proposition}[theorem]{Proposition}

\renewcommand\appendix{\par
  \setcounter{section}{0}
  \setcounter{subsection}{0}
  \setcounter{figure}{0}
  \setcounter{table}{0}
  \renewcommand\thesection{\Alph{section}}
  \renewcommand\thefigure{\Alph{section}\arabic{figure}}
  \renewcommand\thetable{\Alph{section}\arabic{table}}
}

\newcommand{\C}{\mathrm{C}}
\newcommand{\IM}{\mathrm{IM}}

\newcommand{\Sym}{\mathrm{Sym}}

\newcommand{\nats}{\mathbb{N}}
\newcommand{\pr}{\mathrm{pr}}
\newcommand{\im}{\mathrm{Im}}

\newcommand{\Sol}{\mathrm{S}}
\newcommand{\PC}{\mathrm{PC}}
\newcommand{\MC}{\mathrm{MC}}
\newcommand{\NC}{\mathrm{NC}}

\begin{document}
	
\title{\textbf{On the relative power of linear algebraic approximations of graph isomorphism
}}
\author{Anuj Dawar \and Danny Vagnozzi}
\maketitle

\begin{abstract}
We compare the capabilities of two approaches to approximating graph isomorphism using linear algebraic methods: 
the \emph{invertible map tests} (introduced by Dawar and Holm in \cite{dho}) and  proof systems with algebraic rules, namely \emph{polynomial calculus}, \emph{monomial calculus} and \emph{Nullstellensatz calculus}.  In the case of fields of characteristic zero, these variants are all essentially equivalent to the the Weisfeiler-Leman algorithms.
In positive characteristic we show that the invertible map method can simulate the monomial calculus and identify a potential way to extend this to the monomial calculus.
\end{abstract}

\section{Introduction}

The \emph{graph isomorphism problem} consists in deciding whether there is an edge-preserving bijection between the vertex sets of two given graphs. Computationally, this problem is polynomial-time equivalent to finding the partition into orbits of the induced action of the automorphism group of a given graph on a fixed power $k$ of its vertex set \cite{mathon} (we shall refer to this partition as the $k$-\emph{orbit partitions} for a graph). The complexity of these problems is notoriously unresolved: while there are reasons to believe that they are not NP-complete, it is still an open problem as to whether they are in P. The best known upper bound to their computational time is quasi-polynomial, which follows from a breakthrough by Babai \cite{babai}.

There has been a surge of interest recently in linear-algebraic approaches to the graph isomorphism problem.  This is partly because the examples which give hard instances to standard graph isomorphism algorithms code systems of linear equations over finite fields and partly because linear-algebraic groups seem to be the barrier to improving Babai's algorithm.  In this paper, we consider two distinct methods for incorporating algorithms for solving linear systems into graph isomorphism solvers and compare them.  The first is based based on the use of algebraic proofs systems, such as the polynomial calculus and the second are generalizations of the Weisfeiler-Leman method, based on stability conditions and coherent algebras.  We next look at these two approaches more closely.

Let $\Gamma$ and $\Gamma'$ be graphs with vertex set $V$ and adjacency matrices $A$ and $A'$ respectively. Then, $\Gamma$ and $\Gamma'$ are isomorphic if, and only if, there is a $V\times V$ permutation matrix $X$ such that $AX=XA'$. The latter matrix equation can be viewed as a system of linear equations in the variables $x_{vw}$ which are the entries of $X$ for $v,w\in V$. If  we add the conditions that all rows and columns of $X$ must sum to one, and all variables must be non-negative integers, then such a linear program has an integer solution if, and only if, $\Gamma$ and $\Gamma'$ are isomorphic. The feasibility of a linear program over the integers is, however, not known to be decidable in polynomial time, so it is natural to investigate what happens when one searches for rational solutions. It has been shown by Tinhofer \cite{tinhofer} that such a solution exists if, and only if, $\Gamma$ and $\Gamma'$ are not distinguished by \emph{na\"{i}ve colour refinement}. The latter is a method consisting of iteratively refining a partition of the set of vertices of a graph until reaching a stable labelled partition, which can be seen as a \emph{canonical colouring} of said graph, approximating the orbits of the action of its automorphism group.

For every $k\in \nats$, the $k$-\emph{Weisfeiler-Leman} algorithm is a  generalization of the na\"{i}ve colour refinement, giving an approximation of the $k$-orbit partition. For a graph with vertex set $V$, each of these algorithms outputs in time $|V|^{O(k)}$ a canonical labelled partition of $V^k$ satisfying a stability condition and respecting local isomorphism. Informally, they can be seen as forming a family of algorithms, each defining a notion of equivalence on both graphs and tuples of vertices of graphs. A result by Cai, F\"{u}rer and Immerman \cite{cfi} shows how to construct graphs $\Gamma_k$ of size $O(k)$ for which the $k$-Weisfeiler-Leman algorithm fails to produce the $k$-orbit partition. In the same paper, it is shown that the equivalence classes of the output partition of the $k$-Weisfeiler-Leman algorithm coincide with the equivalence classes of $k$-tuples of vertices distinguished by \emph{counting logic} formulae with at most $k+1$ variables. Thus, one deduces from the tight connection made by Immerman and Lander (see Theorem 2.3 in \cite{daw1}), that the equivalence notions defined by the Weisfeiler-Leman family of algorithms delimit the expressive power of \emph{fixed point logic with counting} (FPC). Intuitively, one such limitation is the expressibility of solvability of systems of linear equations over finite fields, since the above mentioned constructions by Cai, F\"{u}rer and Immerman are essentially graph encodings of systems of linear equations over $\mathbb{Z}_2$ \cite{ats}. This has therefore prompted research into families of algorithms graded by the natural numbers, whose notion of equivalence on tuples of vertices of graphs is conceptually a linear algebraic invariance over some field $\mathbb{F}$. One such family is that of the \emph{invertible map tests} over $\mathbb{F}$, first defined in \cite{dho}. For any graph, the $k^{th}$ algorithm of this family also produces a canonical labelled partition of $k$-tuples of its vertices, satisfying a stability condition and respecting local isomorphism, thus giving another notion of equivalence on both graphs and $k$-tuples of vertices thereof. Furthermore, such equivalences delimit the expressive power of the extension of fixed-point logic with quantifiers over linear algebraic operators as shown in \cite{dgp}.

For a fixed characteristic, the output of the $k$-invertible map tests is independent of the field; as such, $\mathbb{F}$ will hereafter be taken to be a prime field without loss of generality. Informally, one can claim that the family of equivalences on tuples defined by the Weisfeiler-Leman algorithms and that defined by the invertible map tests over $\mathbb{Q}$ \emph{simulate} each other in the following sense. For every $k\in \nats$, there is some $k'\in \nats$ such that for any graph, any pair of $k$-tuples of its vertices distinguished by the $k$-Weisfeiler-Leman algorithm are distinguished by the $k'$-invertible map test over $\mathbb{Q}$. Conversely, for every $k\in \nats$ there is a $k'\in \nats$ such that for any graph, any pair of $k$-tuples of vertices distinguished by the $k$-invertible map test over $\mathbb{Q}$ is distinguished by the $k'$-Weisfeiler-Leman algorithm. If the characteristic of $\mathbb{F}$ is positive the former statement holds, but not the latter; indeed, it is shown in \cite{holm} and \cite{dgp} how one can construct graphs $\Gamma_{k,p}$ for each $k\in \nats$ and prime number $p$, for which the $3$-invertible map test over $\mathbb{Z}_p$ outputs the $3$-orbit partition, but the output of the $k$-invertible map test over $\mathbb{Z}_q$ with $q\neq p$ is strictly coarser than the $k$-orbit partition.

The second approach to approximating the orbit partition, relying on linear algebraic operations that we consider is that of algebraic proof systems~\cite{imp,clegg}.  These systems are the subject of very active study in the area of proof complexity.  They have been studied specifically in the context of graph isomorphism in~\cite{berkh} and~\cite{g2p2}.  In particular, the proof systems studied are the  \emph{polynomial calculus} (PC), and the weaker \emph{Nullstellensatz calculus} (NC) and \emph{monomial calculus} (MC).  Each of these gives, for a fixed field $\mathbb{F}$ a  set of rules $\mathcal{R}$ dictating how new polynomials, with coefficients in $\mathbb{F}$, may be derived from an initial set of polynomials, which we shall refer to as \emph{axioms}.   For a graph $\Gamma$ on $V$, we say that the corresponding calculus distinguishes $\vec{u},\vec{v}\in V^k$ if there is a $\mathcal{R}$-derivation of $x_{u_1v_1}x_{u_2v_2}\hdots x_{u_kv_k}$ from $\mathrm{Ax}(\Gamma)\subset\mathbb{F}[\{x_{uv}| u,v\in V\}]$, a set of axioms encoding the atomic types of the input graph $\Gamma$.  Say that a derivation has \emph{degree} $d$ if all polynomials occurring in the derivation have degree at most $d$.  For each of the calculi, fixing the degree $d$ gives us a polynomial-time algorithm for checking the existence of a derivation and hence a polynomial-time approximation of the orbit partition on graphs.  Again, we may restrict $\mathbb{F}$ to a prime field without loss of generality---this will become clear once we define properly the set $\mathrm{Ax}(\Gamma)$ and the derivation rules in Section~\ref{calculi}.

The question we address in this paper is how the approximations of the orbit partition obtained by these algebraic proof systems for a given field $\mathbb{F}$ compare with those we get from the invertible map test.  In the case of fields of characteristic zero, the answer is quite clear, as both approaches yield algorithms that are (up to constant factors) equivalent to the Weisfeiler-Leman algorithms.  This is shown for the invertible maps in~\cite{addv} and for the polynomial calculus in~\cite{g2p2}, giving the following statement.
\begin{theorem}
  For any $k\in \nats$ and $\vec{u},\vec{v}\in V^k$, there is a $k'\in \nats$ such that if $x_{u_1v_1}\hdots x_{u_kv_k}$ has a degree $k$ PC derivation over $\mathbb{Q}$ from $\mathrm{Ax}(\Gamma)$, then $\vec{u}$ and $\vec{v}$ are distinguished by the $k'$-invertible map test over $\mathbb{Q}$. Conversely, for any $k\in \nats$ and $\vec{u},\vec{v}\in V^k$, there is a $k'\in \nats$ such that if $\vec{u}$ and $\vec{v}$ are distinguished by the $k$-invertible map test over $\mathbb{Q}$, then $x_{u_1v_1}\hdots x_{u_kv_k}$ has a degree $k'$ PC derivation over $\mathbb{Q}$ from $\mathrm{Ax}(\Gamma)$.
\end{theorem}

In the case of positive characteristic, we show (in Section~\ref{sec:monomial}) the definability of derivations of MC in FPS$(p)$, an extension of fixed-point logic with quantifiers for the solvability of systems of linear equations over fields of characteristic $p$.  This implies, in particular, that the approximation of the orbit partition obtained by MC in characteristic $p$ is no finer than that obtained by the invertible map test in characteristic $p$.
\begin{theorem}\label{thm:main1}
  For any prime number $p$, $k\in \nats$ and $\vec{u},\vec{v}\in V^k$, there is a $k'\in \nats$ such that if $x_{u_1v_1}\hdots x_{u_kv_k}$ has a degree $k$ MC derivation over $\mathbb{Z}_p$ from $\mathrm{Ax}(\Gamma)$, then $\vec{u}$ and $\vec{v}$ are distinguished by the $k'$-invertible map test over $\mathbb{Z}_p$.
\end{theorem}

In the other direction, we are able to show that the distinguishing power of FPS$(p)$ can be simulated by NC.  Thus, this apparently weaker proof system is able to simulate (as far as the graph isomorphism problem is concerned) PC in characteristic zero and at least MC in positive characteristic.
\begin{theorem}\label{thm:main2}
\begin{enumerate}
\item For any $k\in \nats$ and $\vec{u},\vec{v}\in V^k$, there is a $k'\in \nats$ such that if $x_{u_1v_1}\hdots x_{u_kv_k}$ has a degree $k$ PC derivation of over $\mathbb{Q}$ from $\mathrm{Ax}(\Gamma)$, then there is also a degree $k'$ NC derivation over $\mathbb{Q}$ from $\mathrm{Ax}(\Gamma)$.
\item For any prime number $p$, $k\in \nats$ and $\vec{u},\vec{v}\in V^k$, there is a $k'\in \nats$ such that if $x_{u_1v_1}\hdots x_{u_kv_k}$ has a degree $k$ MC derivation over $\mathbb{Z}_p$ from $\mathrm{Ax}^{(k)}(\Gamma)$, then it also has a degree $k'$ NC derivation over $\mathbb{Z}_p$ from $\mathrm{Ax}(\Gamma)$.
\end{enumerate}
\end{theorem}

 From this, a strengthening of Theorem 6.3 in \cite{berkh} also follows. Let $V$ be the vertex set of $\Gamma_{k,p}$ as above.
\begin{theorem}\label{thm:main3}
If $\vec{u},\vec{v}\in V^3$ are not in the same equivalence class of the $3$-orbit partition of $\Gamma_{k,p}$, then $x_{u_1v_1}x_{u_2v_2}x_{u_3v_3}$ has a degree $3$ NC derivation over $\mathbb{Z}_p$ from $\mathrm{Ax}(\Gamma_{k,p})$.
\end{theorem}

\paragraph{Structure of the paper.} We first give an overview of the relevant background (Sections~\ref{sec:prelim} and~\ref{sec:operators}).  Section~\ref{sec:monomial} is mostly concerned with the definability of monomial calculus refutations over $\mathbb{Z}_p$ in FPS$(p)$ and contains the proof of Theorem \ref{thm:main1}. Theorems \ref{thm:main2} and \ref{thm:main3} are proven in Section~\ref{sec:nullstellensatz}.  We conclude by discussing related open problems and ways to approach them.

\section{Preliminaries}\label{sec:prelim}
In this section, after introducing some notational conventions, we
recall some notions on partitions introduced in Sections 3 and 4 of
\cite{addv} and briefly describe the logics of interest.

\paragraph{Notational conventions.} All sets are finite unless stated otherwise. Given two sets $V$ and $I$, a tuple in $V^I$ is denoted by $\vec{v}$, and its $i^{th}$ entry by $v_i$, for each $i\in I$. We use the notation $(v_i)$ for $i\in I$ to denote the element of $V^I$ with $i^{th}$ element equal to $v_i$. For a field $\mathbb{F}$, the element of $\mathbb{F}^V$ with all its entries equal to $1$ is denoted by $\mathbf{1}_V$. We set $[k]=\{1,2,\hdots,k\}\subset\mathbb{N}$ and define $[k]^{(r)}=\{\vec{x}\in [k]^r \mid x_i\neq x_j\;\forall i,j\in [r],\;i\neq j\}$ for $r\leq k$.

For any $\vec{v}\in V^k$, $\vec{u}\in V^r$, and $\vec{i}\in [k]^{(r)}$ we define $\vec{v}\langle \vec{i},\vec{u}\rangle\in V^k$ to be the tuple with entries
$$
(\vec{v}\langle \vec{i},\vec{u}\rangle)_j=\begin{cases}
u_{i_s}\;\,\text{if $j=i_s$ for some $s\in [r]$}\\
v_{j}\;\,\text{otherwise.}\\
\end{cases}
$$
In other words, $\vec{v}\langle \vec{i},\vec{u}\rangle$ is the tuple obtained from $\vec{v}$ by substituting the elements of $\vec{u}$ in the positions specified by $\vec{i}$. 

Given two tuples $\vec{v}\in V^r$ and $\vec{w}\in V^s$, their \emph{concatenation} is denoted by $\vec{v}\cdot \vec{w}\in V^{r+s}$. More precisely, $\vec{v}\cdot\vec{w}$ is the tuple with entries
$$
(\vec{v}\cdot\vec{w})_i=\begin{cases}
v_i \;\,& \textrm{if $i\in [r]$}\\
w_j\;\,& \textrm{if $i=j+r$}.
\end{cases}
$$

For a relation $R\subseteq V^2$ we define the \emph{adjacency matrix} of $R$ to be the $V\times V$ matrix whose $(u,v)$ entry is $1$ if $(u,v)\in R$ and $0$ otherwise.

\subsection{Labelled partitions and refinement operators}

A \emph{labelled partition} of a set $A$ is a map $\gamma:A\rightarrow
X$, where the elements of $X$ are sometimes referred to as
\emph{colours}, and $\gamma(a)$ as the colour of $a$ in $\gamma$.
Denote the class of labelled partitions of $A$ by $\mathcal{P}(A)$.
For equivalence relations $\mathfrak{R}$ and $\mathfrak{S}$ on $A$ we
write $\mathfrak{R}\preceq _A \mathfrak{S}$ and say that
$\mathfrak{S}$ \emph{refines} $\mathfrak{R}$ if, whenever $a,b \in A$
are in the same equivalence class of $\mathfrak{S}$, they are also in
the same equivalence class of $\mathfrak{R}$. We extend the partial
order $\preceq_A$ to labelled paritions by writing $\gamma \preceq_A
\rho$ to mean that the equivalence relation $\{ (a,b) \mid\ \rho(a) =
\rho(b)\}$ refines the relation $\{(a,b)\mid \gamma(a) = \gamma(b)\}$.  Note that this does not require that the co-domains of $\gamma$ and $\rho$ are the same or indeed related in any way.  We omit the subscript $A$ where the set is clear from the context. 

Define an action of $\Sym(k)$ on $V^k$ by setting $\vec{v}^\pi$ to be the element of $V^k$ with $i^{th}$ entry $v_{\pi^{-1}(i)}$. $\gamma\in\mathcal{P}(V^k)$ is said to be \emph{invariant} if $\gamma(\vec{u})= \gamma(\vec{v})$ implies $\gamma(\vec{u}^\pi) = \gamma(\vec{v}^\pi)$ for all $\pi\in\Sym(k)$ and $\vec{u},\vec{v}\in V^k$. For $t\in[k]$, the $t$-\emph{projection} of $\gamma$ is defined to be the labelled partition $\pr_t\gamma\in\mathcal{P}(V^t)$ such that for all $\vec{u}\in V^t$
$$
\pr_t\gamma(\vec{u})=\gamma(u_1,u_2,\hdots,u_t,\hdots,u_t).
$$
For $\vec{v}\in V^k$, with $t,k$ as above, we similarly define $\pr_t\vec{v}$ to be the $t$-tuple $(v_1,v_2,\hdots,v_t)$.

For a labelled partition $\gamma\in\mathcal{P}(V^k)$ and $t\in [k]$, we say that $\vec{u},\vec{v}\in V^t$ are \emph{distinguished} by $\gamma$ if $\pr_t\gamma(\vec{u})\neq\pr_t\gamma(\vec{v})$.

\begin{definition}[Graph-like partition\footnote{The definition of a graph-like partition in \cite{addv} has an additional assumption (that of $k$-\emph{consistency}). One can show, however, that the results in said paper still hold true without the $k$-consistency assumption.}]
$\gamma$ is said to be \emph{graph-like} if it is invariant and for all $i,j\in[k]$
$$\gamma(\vec{u})=\gamma(\vec{v})\implies (u_i=u_j\implies v_i=v_j).$$
\end{definition}
In other words, a partition of $V^k$ is graph-like if it refines the
partition into equality types.

For $\vec{i}\in [k]^{(2r)}$ and $\vec{v}\in V^k$, define the $(\vec{i},\vec{v})$-\emph{character vector} of $\gamma\in\mathcal{P}(V^k)$ to be the tuple $\vec{\chi}=(\chi_\sigma)_{\sigma\in\im(\gamma)}$ where $\chi_\sigma$ is the adjacency matrix of the relation $\{(\vec{x},\vec{y})\in (V^{r})^2 \mid \gamma(\vec{v}\langle\vec{i},\vec{x}\cdot\vec{y}\rangle)=\sigma\}$.
\begin{definition}[Refinement operator]
A $k$-\emph{refinement operator} is a map $R$ which, for any set $V$, assigns to each $\gamma\in\mathcal{P}(V^k)$ a partition $R\circ\gamma\in\mathcal{P}(V^k)$ that satisfies the following properties:
\begin{itemize}
	\item[(R1).] Refinement: $\gamma\preceq R\circ\gamma$.
	\item[(R2).] Monotonicity: $\gamma\preceq\rho\implies R\circ\gamma\preceq R\circ\rho$.
\end{itemize}
\end{definition}

We say that a partition $\gamma \in \mathcal{P}(V^k)$ is
\emph{$R$-stable} if $R\circ\gamma = \gamma$.
Given an $X \in\mathcal{P}(V^k)$, define a sequence of labelled
partitions by $X^0=X$ and $X^{i+1}=R\circ X^i$. Then, there is some
$s$ such that for all $i$,  $i\geq s$ implies that $X^i$ is $R$-stable.  For the minimal such $s$ we denote $X^s$ by $[X]^R$. By the monotonocity property of $R$, it follows that up to re-labelling of its equivalence classes, $[X]^R$ is the unique minimal $R$-stable partition refining $X$.

A $k$-refinement operator $R$ is said to be \emph{graph-like} if $R\circ\gamma$ is graph-like for any graph-like $\gamma$.

\subsection{Extensions of first order and inflationary fixed point logics}
We assume the reader has some familiarity with first-order and
fixed-point logics, whose details can be found, for instance
in~\cite{EbbFlum}. Throughout the paper, for a logic $\mathcal{L}$, we
denote by $\mathcal{L}_k$ the class of all $\mathcal{L}$-formulae
(over some pre-specified vocabulary) with at most $k$ variables. Let
$\mathfrak{A}$ be a structure with universe $V$ and fix
$\vec{u},\vec{v}\in V^k$. We say that some $\mathcal{L}_k$ formula
$\phi(\vec{z})$ \emph{distinguishes}
$(\mathfrak{A},\vec{z}\mapsto\vec{u})$ from
$(\mathfrak{A},\vec{z}\mapsto\vec{v})$ if either
$\mathfrak{A}\models\phi(\vec{u})$ and
$\mathfrak{A}\not\models\phi(\vec{v})$ or $\mathfrak{A}\not\models\phi(\vec{u})$ and
$\mathfrak{A}\models\phi(\vec{v})$ .

\paragraph{A language for graphs.} We may view any graph as an
\emph{edge-coloured complete digraph} and hence, as a labelled
partition of the set of ordered pairs of its vertices. For example, a
loop-free uncoloured graph on $V$ is a partition of $V^2$ whose colour
set is $\{\texttt{vertex, edge, non-edge}\}$. From a logical
view-point, we may thus define a graph $\Gamma$ on $V$ to be a finite
relational structure $\mathfrak{A}_\Gamma$ with universe $V$, over a
vocabulary consisting of binary relation symbols.

\paragraph{Counting logic.} \emph{Counting logic} $\mathcal{C}$ is the
extension of first-order logic with \emph{counting quantifiers}
$\exists^m$, for each positive integer $m$. A structure satisfies the
formula $\exists^mx.\phi(x)$ if, and only if, there are at least $m$ elements in the universe for which the formula $\phi(x)$ is satisfied. Note that $\exists^mx.\phi(x)$ can be written as an equivalent first-order formula by using $m$ quantifiers $\exists$ and $m$ auxiliary variables. This, however, inconveniently increases the number of variables required.

\paragraph{Solvability quantifiers.} For a logic $\mathcal{L}$ we
denote its extension via \emph{solvability quantifiers}
$\mathrm{slv}_p$ over a finite field of characteristic $p$ by
$\mathcal{L}+$S$(p)$. Let $\phi(\vec{x},\vec{y},\vec{z})$ be a formula
of such a logic, where $\vec{x},\vec{y},\vec{z}$ are $i,j,k$-tuples of
variables respectively. Then
$\mathrm{slv}_p(\vec{x}\vec{y}.\phi(\vec{x},\vec{y},\vec{z}))$ is also
a $\mathcal{L}+$S$(p)$ formula.  See~\cite{defin,g2p2} for more about
these quantifiers.

The semantics of this quantifier is as follows. To each structure with universe $V$ and $k$-tuple $\vec{v}\in V^k$, we associate the $V^i\times V^j$ matrix $S_\phi^{\vec{v}}$ over $\{0,1\}\subseteq \mathbb{Z}_p$ with $(\vec{r},\vec{s})$-entry equal to $1$ if, and only if, $\mathfrak{A}\models\phi(\vec{r},\vec{s},\vec{v})$. Then, $(\mathfrak{A},\vec{z}\mapsto\vec{v})\models\mathrm{slv}_p(\vec{x}\vec{y}.\phi(\vec{x},\vec{y},\vec{z}))$ if, and only if, there is some $\vec{a}\in \mathbb{Z}_p^{V^i}$ such that $S^{\vec{v}}_\phi\vec{a}=\mathbf{1}$.

Note that the fact that the linear equations encoded by the solvability quantifier may seem to have a restricted form does not entail any loss of generality, as is explained in Section \ref{solsec}. Furthermore, we may take $i=j$ in the above (and hence, assume that $S_\phi^{\vec{v}}$ is a square matrix), by simply adding auxiliary rows or columns. For a more detailed explanation see Section 3.2 of \cite{dgp}, which deals with quantifiers over linear algebraic operators, of which the solvability quantifier is a special case.

When $\mathcal{L}$ is fixed-point logic or first order logic, we
denote $\mathcal{L}+$S$(p)$ by FPS$(p)$ or FOS$(p)$ respectively,
and FOS$_k(p)$ denotes the corresponding logic with formulae with at
most $k$ variables.

\paragraph{Embedding fixed-points in finite variable logics.} Fixed
point logics provide a compact framework for describing sets which are
defined recursively.  For the purpose of establishing limits on their
expressive power, it is easier to study them in terms of finite
variable logics. which provide bounds for their expressive power.  For
example, it is well known that if $\phi$ is a first-order formula in
$m$ variables, then $\mathrm{ifp}_{R,\vec{x}}\phi$ is equivalent, on
structures with at most $n$ elements (for arbitrary $n$) to a
first-order formula with $2m$ variables; see Section 11.1 in
\cite{libkin} for a detailed argument.  It follows that
$\mathrm{ifp}_{R,\vec{x}}\phi$ cannot distinguish between two
structures which are indistinguishable by any first-order formula with
at most $2m$ variables.
It is easy to see that a similar argument yields an analogous result for FPS$(p)$.  That is if $\phi$ is an FPS$(p)$ formula with at most $m$ variables, it cannot distinguish any pair of structures which agree on all FOS$(p)$ formulas with at most $2m$ variables.

\paragraph{Logical interpretations.} A logical interpretation $\mathcal{I}$ maps a structure $\mathfrak{A}$ to a structure $\mathcal{I}(\mathfrak{A})$ according to a rule dictated by a tuple of formulae from a certain logic $\mathcal{L}$. They can be seen as the logical counterpart of an algorithmic reduction. Formally, let $\mathtt{S}$ and $\mathtt{S}'$ be relational vocabularies, and let $\mathtt{T}_1\hdots\mathtt{T}_l$ be the symbols of $\mathtt{S}'$, with $t_i$ indicating the arity of $\mathtt{T}_i$. An $\mathcal{L}[\mathtt{S},\mathtt{S}']$-\emph{interpretation} is a tuple
$$
\mathcal{I}(\vec{z})=(\phi_\delta(\vec{x},\vec{z}),\phi_\equiv(\vec{x},\vec{y},\vec{z}),\phi_{\mathtt{T}_1}(\vec{x}_1,\hdots,\vec{x}_{t_1},\vec{z}),\hdots,\phi_{\mathtt{T}_l}(\vec{x}_1,\hdots,\vec{x}_{t_l},\vec{z}))
$$
where $\phi_\delta,\phi_\equiv,\phi_{\mathtt{T}_1},\hdots,\phi_{\mathtt{T}_l}$ are $\mathcal{L}$-formulae over $\mathtt{S}$, $\vec{x},\vec{y},\vec{x}_1,\hdots,\vec{x}_{t_l}$ are tuples of pairwise distinct variables of length $d$ and $\vec{z}$ is a tuple of variables pairwise distinct from the previous ones. We refer to $d$ and $\vec{z}$ as the \emph{arity} and the \emph{parameters} of the interpretation respectively, and to $\phi_{\equiv}$ as the \emph{equivalence formula}.

With $\mathcal{I}(\vec{z})$ a $d$-ary interpretation as above, let $\mathfrak{A}$ be a $\mathtt{S}$-structure with universe $V$, and $\vec{a}$ a tuple of elements of $V$ with same length as $\vec{z}$. Define a $\mathtt{S}'$-structure $\mathfrak{B}$ with universe $\{\vec{b}\in V^d\mid \mathfrak{A}\models\phi_\delta(\vec{b},\vec{z})\}$, and set $\mathfrak{B}\models \mathtt{T}_i(\vec{b}_1,\hdots,\vec{b}_{t_i})$ if, and only if, $\mathfrak{A}\models\phi_{\mathtt{T}_i}(\vec{b}_1,\hdots,\vec{b}_{t_i},\vec{a})$. Moreover, let $\mathcal{E}=\{(\vec{b}_1,\vec{b}_2)\in V^d\times V^d\mid \mathfrak{A}\models\phi_{\equiv}(\vec{b}_1,\vec{b}_2,\vec{a})\}$. We then define
$$
\mathcal{I}(\mathfrak{A},\vec{z}\mapsto\vec{a})=\begin{cases}
\mathfrak{B}/\mathcal{E}\;\;\textrm{ if $\mathcal{E}$ is a congruence relation on $\mathfrak{B}$}\\
\textrm{undefined otherwise}.
\end{cases}
$$

Hence, intrepretations can be seen to induce maps from $\mathtt{S}$-structures to $\mathtt{S}'$-structures. Their crucial property is that, in terms of $\mathcal{L}$-formulae, they can be reversed to give a map from $\mathcal{L}$-formulae over $\mathtt{S}'$ to $\mathcal{L}$-formulae over $\mathtt{S}$ as the next result illustrates. A proof can be found in most standard textbooks (see, for example \cite{flum}) with $\mathcal{L}$ being first order logic. Using induction on the structure of $\psi$, it is easy to adapt the proof to $\mathcal{L}$ being fixed point logics and/or its extension with solvability quantifiers.

\begin{lemma}[Interpretation Lemma]
Let $\mathcal{I}(\vec{z})$ be a $d$-ary $\mathcal{L}[\mathtt{S},\mathtt{S}']$ interpretation. Then, for every $\mathcal{L}$-formula $\psi(x_1,\hdots x_m)$ over $\mathtt{S}'$, there is a $\mathcal{L}$-formula $\psi^{-\mathcal{I}}(\vec{y}_1,\hdots,\vec{y}_m,\vec{z})$ over $\mathtt{S}$ (where $\vec{y}_i$ are $d$-tuples of pairwise distinct variables) such that for all $(\mathfrak{A},\vec{z}\mapsto\vec{a})$ with $\mathcal{I}(\mathfrak{A},\vec{z}\mapsto\vec{a})$ well defined, and for all $d$-tuples $\vec{q}_1,\hdots,\vec{q}_m$ of elements from the universe of $\mathfrak{A}$ such that $\vec{q}_1/_\equiv,\hdots,\vec{q}_m/_\equiv$ are elements of the universe of $\mathcal{I}(\mathfrak{A},\vec{z}\mapsto\vec{a})$
$$
\mathfrak{A}\models\psi^{-\mathcal{I}}(\vec{q}_1,\hdots,\vec{q}_m,\vec{a})\iff \mathcal{I}(\mathfrak{A},\vec{z}\mapsto\vec{a})\models\psi(\vec{q}_1/_\equiv,\hdots,\vec{q}_m/_\equiv).
$$
\end{lemma}

\section{Refinement operators and proof systems with algebraic rules}\label{sec:operators}
In this section we give an overview of the refinement operators of interest and briefly describe the proof systems mentioned in the introduction. In particular, we recall the invertible map and counting logics operators and introduce the solvability operators.
\subsection{The invertible map tests}
The equivalence relations on tuples of vertices induced by the
invertible map tests have been originally introduced in \cite{dho},
under the guise of a pebble game with algebraic rules on a pair of
graphs.  In \cite{dgp}, it was shown that these equivalences
correspond to indistinguishability in an extension of first order
logic with linear algebraic operators (Theorem~2).  Algorithmically,
one can find these equivalence classes by computing a fixed point of
the \emph{invertible map operators} $\IM_{k,r}^\mathbb{F}$, as defined
in Sections~8 of \cite{addv}.  Each of these is a $k$-refinement
operator such that for $\gamma\in\mathcal{P}(V^k)$, the colour of
$\vec{v}\in V^k$ in the partition $\IM_{k,r}^\mathbb{F}\circ\gamma$ is
given by a tuple whose entries are $\gamma(\vec{v})$ and the
equivalence classes under matrix conjugation of the
$(\vec{i},\vec{v})$-character vectors of $\gamma$, for all
$\vec{i}\in[k]^{(2r)}$. We shall not need the definitions of these
operators, as their \emph{stability conditions} are simpler to state and more useful for our purpose. That is, $\gamma$ is $\IM_{k,r}^{\mathbb{F}}$-stable if for all $\vec{u},\vec{v}\in V^k$ and $\vec{i}\in[k]^{(2r)}$
$$
\gamma(\vec{u})=\gamma(\vec{v})\implies \exists M\in \mathrm{GL}_{V^r}(\mathbb{F}),M\chi_{\sigma}M^{-1}=\xi_{\sigma}\,\,\forall \sigma\in\im(\gamma)
$$
where $(\chi_\sigma)_{\sigma\in\im(\gamma)}$ and $(\xi_\sigma)_{\sigma\in\im(\gamma)}$ are the $(\vec{i},\vec{v})$ and $(\vec{i},\vec{u})$-character vectors of $\gamma$ respectively.

For a graph $\Gamma$ on $V$, let $\alpha_{k,\Gamma}$ be a canonical labelled partition of $V^k$ into atomic types of $\Gamma$.\footnote{By \emph{canonical} we mean \emph{invariant under isomorphism}. That is, if $\Gamma$ and $\Gamma'$ are graphs on $V$ and $V'$ respectively and $\vec{u}\in V^k$ and $\vec{v}\in (V')^k$, then $\alpha_{k,\Gamma}(\vec{u})=\alpha_{k,\Gamma'}(\vec{v})$ if, and only if, the mapping $u_i\rightarrow v_i$ is an isomorphism of the subgraphs induced by the vertices in $\vec{u}$ and $\vec{v}$.}
Define the $k$-refinement operator $\IM_k^\mathbb{F}$ so that for $\gamma\in\mathcal{P}(V^k)$, the colour of $\vec{v}\in V^k$ in $\IM_k^\mathbb{F}\circ\gamma$ is given by a tuple whose entries are the colours of $\vec{v}$ in $\IM_{k,r}^\mathbb{F}\circ\gamma$, for all $r<k/2$. Formally,
$$
\IM_k^\mathbb{F}\circ\gamma(\vec{v})=(\IM_{k,1}^\mathbb{F}\circ\gamma(\vec{v}),\IM_{k,2}^\mathbb{F}\circ\gamma(\vec{v}),\hdots,\IM_{k,\lfloor k/2\rfloor}^\mathbb{F}\circ\gamma(\vec{v})).
$$
Then, the output of the $k$-invertible map test over $\mathbb{F}$ is the labelled partition $[\alpha_{k,\Gamma}]^{\IM_k^\mathbb{F}}$. It is explained in Proposition $4.6$ in \cite{addv} how one can obtain this partition in time $|V|^{O(k)}$ by iteratively applying $\IM_k^\mathbb{F}$ to $\alpha_{k,\Gamma}$. For $\vec{u},\vec{v}\in V^l$ and $l\leq k$, we say that the $k$-invertible map over $\mathbb{F}$ \emph{distinguishes} $\vec{u}$ and $\vec{v}$ if they are in different equivalence classes of the partition $\pr_l[\alpha_{k,\Gamma}]^{\IM_k^\mathbb{F}}$.

Finally, we note that for a fixed characteristic, the choice of $\mathbb{F}$ is irrelevant: indeed the matrices $A_1,\hdots,A_m\in\mathrm{Mat}_V(\mathbb{F})$ and $B_1,\hdots,B_m\in\mathrm{Mat}_V(\mathbb{F})$ are simultaneously similar over $\mathbb{F}$ if, and only if, they are simultaneously similar over any extension of $\mathbb{F}$ \cite{pazzis}. Since the entries of the character vectors are $01$-matrices, we may assume, without loss of generality, that $\mathbb{F}$ is a prime field. Hereafter, we shall then indicate the operators $\IM_{k,r}^\mathbb{F}$ and $\IM_k^\mathbb{F}$ by $\IM_{k,r}^c$ and $\IM_k^c$ respectively, where $c$ is the characteristic of $\mathbb{F}$. 

\subsection{Counting logics operators}
It is useful to express the partition of $k$-tuples into equivalence
classes under finite variable counting logics as the fixed point of a
refinement operator. For this purpose, the $k$-refinement operators
$\C_{k,r}$, for $r<k$, have been defined so that for any
$\gamma\in\mathcal{P}(V^k)$, the colour of $\vec{v}\in V^k$ in
$\C_{k,r}\circ\gamma$ is given by a tuple whose entries are
$\gamma(\vec{v})$ and the multisets of colours in $\gamma$ of the
tuples which can be obtained by substituting an $r$-tuple into
$\vec{v}$ (see Section 4 of \cite{addv}). In particular, $\gamma$ is
$\C_{k,r}$-stable if, and only if, for all $\vec{i}\in [k]^{(r)}$ and
$\sigma\in\im(\gamma)$, the size of the set $\{\vec{x}\in V^r \mid
\gamma(\vec{v}\langle\vec{i},\vec{x}\rangle)=\sigma\}$ is independent
of the choice of $\vec{v}$ from within its equivalence class in
$\gamma$. As such, for a graph $\Gamma$ on $V$, $\vec{u},\vec{v}\in
V^k$ are in the same equivalence class of
$[\alpha_{k,\Gamma}]^{\C_{k,1}}$ if, and only if, there are no
$\mathcal{C}_k$ formulae distinguishing
$(\mathfrak{A}_\Gamma,\vec{z}\mapsto\vec{u})$ from
$(\mathfrak{A}_\Gamma,\vec{z}\mapsto\vec{v})$. 
The combinatorial properties of $\C_{k,r}$-stable partitions can be
used to show the relation between the distinguishing powers of finite
variable fragments of counting logics and the invertible map tests.
In short, the invertible map test over fields of characteristic zero
is not more distinguishing than counting logic, but over fields of
positive characteristic it is.  To be precise, 
 with $\Gamma,\vec{u},\vec{v}$ as above, the following holds:
\begin{displayquote}
For any field $\mathbb{F}$, if the $k$-invertible map test over $\mathbb{F}$ does not distinguish $\vec{u}$ and $\vec{v}$ in $\Gamma$, then there are no $ \mathcal{C}_{k-1}$-formulae distinguishing $(\mathfrak{A}_\Gamma,\vec{z}\mapsto\vec{u})$ from $(\mathfrak{A}_\Gamma,\vec{z}\mapsto\vec{v})$.

If the $k$-invertible map test over $\mathbb{Q}$ distinguishes $\vec{u}$ from $\vec{v}$ in $\Gamma$, then there is some $\mathcal{C}_{2k-1}$-formula distinguishing $(\mathfrak{A}_\Gamma,\vec{z}\mapsto\vec{u})$ from $(\mathfrak{A}_\Gamma,\vec{z}\mapsto\vec{v})$. \footnote{This is a direct consequence of the following generalizations of Lemmata 7.1 and 7.3 in \cite{addv} respectively: for all $k,r\in \nats$ with $2r<k$, 
\begin{enumerate}
\item The $k$-projection of a graph-like $\IM_{k+r,r}^c$-stable partition is $\C_{k,r}$-stable for any characteristic $c$.
\item The $k$-projection of a graph-like $\C_{k+r,r}$-stable partition is $\IM_{k,r}^0$-stable.
\end{enumerate}
The authors prove it only for the case $r=1$, but a similar argument holds for any $r\in\nats$.}
\end{displayquote}

\subsection{Solvability operators}
Central to this paper are finite variable fragments of the extension of first order logic with solvability quantifiers over a fixed field. In order to construct a refinement operator whose stable points reflect the properties of this logic, we consider a weakened version of the invertible map operators, whose action on labelled partitions can be computed solely by solving systems of linear equations. To achieve this, we proceed by defining an equivalence relation $\sim_{\mathrm{sol}}$ on the character vectors, which can be seen as a relaxation of the conjugation relation $\sim$.

Let $\mathfrak{P}_V(\mathbb{F})=\{A\in \mathrm{Mat}_V(\mathbb{F})
\mid \sum_{w\in V} A_{wv}=\sum_{w\in V}A_{uw}=1, \forall u,v\in V\}$. Or, equivalently, $\mathfrak{P}_V(\mathbb{F})$ is the set of matrices $A\in\mathrm{Mat}_V(\mathbb{F})$ such that $\mathbf{1}_{V}$ is an eigenvector of both $A$ and $A^t$, with corresponding eigenvalue $1$. For a set $I$, let
$$
\mathcal{X}_V^I(\mathbb{F})=\{\vec{A}\in \mathrm{Mat}_V(\mathbb{F})^I \mid \sum_{s\in I}A_s=\mathbb{J}_V\;\mathrm{and}\;\forall i\in I, \exists j, A_i^t=A_j \}.
$$
Note that if $\gamma$ is invariant and $\vec{i}\in [k]^{(2r)}$, then any $(\vec{i},\vec{v})$- character $\vec{\chi}$ of $\gamma$ is an element of $\mathcal{X}_V^{\im(\gamma)}(\mathbb{F})$, since
\begin{equation}\label{sum}
\sum_{\sigma\in\im(\gamma)}\chi_{\sigma}=\mathbb{J}_{V^r}
\end{equation}
and by invariance of $\gamma$, for any $\sigma\in\im(\gamma)$ there is some $\sigma'$ such that $(\chi_\sigma)^t=\chi_{\sigma'}$.

Define the relation $\sim_\mathrm{sol}$ on $\mathcal{X}_V^I(\mathbb{F})$ as follows: $\vec{A}\sim_\mathrm{sol}\vec{B}$ if there is some $S\in \mathfrak{P}_V(\mathbb{F})$ such that $A_iS=SB_i$ for all $i\in I$.

\begin{lemma}
$\sim_\mathrm{sol}$ is an equivalence relation on $\mathcal{X}_V^I(\mathbb{F})$.
\end{lemma}
\begin{proof}
Clearly, $\vec{A}\sim_\mathrm{sol}\vec{A}$, since $\mathbb{I}_V\in \mathfrak{P}_V(\mathbb{F})$ and $A_i\mathbb{I}_V=\mathbb{I}_VA_i$ for all $i\in I$. Suppose $\vec{A}\sim_\mathrm{sol}\vec{B}$. Let $S\in \mathfrak{P}_V(\mathbb{F})$ satisfy $A_iS=SB_i$ for all $i\in I$. Then $B_i^tS^t=S^tA_i^t$ and thus, from the definition of $\mathcal{X}_V^I$, $B_iS^t=S^tA_i$ for all $i\in I$. Since $S^t\in \mathfrak{P}_V(\mathbb{F})$, $\vec{B}\sim_\mathrm{sol}\vec{A}$. Finally, suppose $\vec{A}\sim_\mathrm{sol}\vec{B}$ and $\vec{B}\sim_\mathrm{sol}\vec{C}$. Let $S,T\in \mathfrak{P}_V(\mathbb{F})$ satisfy $A_iS=SB_i$ an $B_iT=TC_i$ for all $i\in I$. Then $A_iST=SB_iT=STC_i$. Since $S,T,S^t$ and $T^t$ must all have $\mathbf{1}_{V^r}$ as eigenvector, with corresponding eigenvalue $1$, so must $ST$ and $(ST)^t$. Hence, $ST\in \mathfrak{P}_V(\mathbb{F})$ and $\vec{A}\sim_\mathrm{sol}\vec{C}$.
\end{proof}

For $k,r\in\nats$ with $2r<k$, a field $\mathbb{F}$, and an invariant  $\gamma\in\mathcal{P}(V^k)$, we define the \emph{solvability operators} $\Sol_{k,r}^\mathbb{F}$ by setting $\Sol_{k,r}^{\mathbb{F}}\circ\gamma$ to be the labelled partition for which the colour of $\vec{v}\in V^k$ is a tuple whose entries are $\gamma(\vec{v})$ and the equivalence classes under the relation $\sim_\mathrm{sol}$ of the $(\vec{i},\vec{v})$-character vectors of $\gamma$, for all $\vec{i}\in[k]^{(2r)}$. Formally:
$$
\begin{matrix}
\Sol_{k,r}^{\mathbb{F}}\circ\gamma: & V^k& \rightarrow & \im(\gamma)\times(\mathcal{X}_{V^r}^{\im(\gamma)}(\mathbb{F})/\sim_\mathrm{sol})^{[k]^{(2r)}}\\
& \vec{v}&\mapsto & (\gamma(\vec{v}),(\vec{\chi}_{\vec{i}})_{\vec{i}\in [k]^{(2r)}}),
\end{matrix}
$$
where $\vec{\chi}_{\vec{i}}$ is the $(\vec{i},\vec{v})$-character vector of $\gamma$.

As before, since the entries of the matrices in the character vector are all $0$ and $1$, we may restrict $\mathbb{F}$ to being a prime field without loss of generality, and denote $\Sol_{k,r}^\mathbb{F}$ by $\Sol_{k,r}^c$, where $c=\mathrm{char}(\mathbb{F})$.

Observe that $\Sol_{k,r}^c$ is monotone on the class of invariant partitions of $V^k$ and is thus a $k$-refinement operator when considered with this domain restriction. Indeed, if $\gamma,\rho\in \mathcal{P}(V^k)$ are invariant and $\gamma\preceq \rho$, then for every $\alpha\in \im(\gamma)$ there is some $I_\alpha\subseteq \im(\rho)$ such that $\gamma^{-1}(\sigma)=\bigcup_{\beta\in I_\sigma} \rho^{-1}(\beta)$. Fix $\vec{i}\in [k]^{(2r)}$ and $\vec{v}\in V^k$, and let $\vec{\chi}$ and $\vec{\xi}$ be the $(\vec{i},\vec{v})$-character vectors of $\gamma$ and $\rho$ respectively. Then
$$
\chi_{\alpha}=\sum_{\beta\in I_{\alpha}}\xi_{\beta}.
$$
It readily follows that if $\Sol_{k,r}^c\circ\rho(\vec{u})=\Sol_{k,r}^c\circ\rho(\vec{v})$, then $\Sol_{k,r}^c\circ\gamma(\vec{u})=\Sol_{k,r}^c\circ\gamma(\vec{v})$ and thus, $\Sol_{k,r}^c\circ\gamma\preceq \Sol_{k,r}^c\circ\rho$.

In a similar fashion to Proposition $4.6$ in \cite{addv}, one can show that $\Sol_{k,r}^c$ are graph-like refinement operators. Thus, for any invariant $\gamma\in\mathcal{P}(V^k)$, $[\gamma]^{\Sol_{k,r}^c}$ is the unique coarsest $\Sol_{k,r}^c$-stable partition refining $\gamma$, up to relabelling of its equivalence classes.

For the remainder of this section, we assume that $\gamma\in\mathcal{P}(V^k)$ is invariant and that $\vec{\chi}$ and $\vec{\xi}$ are the $(\vec{i},\vec{v})$ and $(\vec{i},\vec{u})$-character vectors of $\gamma$ respectively, for some fixed $\vec{i}\in [k]^{(2r)}$. The following is a direct consequence of the definition of $\Sol_{k,r}^c$.
\begin{proposition}
Let $\mathbb{F}$ be the prime field of characteristic $c$. For all $k,r\in \mathbb{N}$, with $2r\leq k$, $\Sol_{k,r}^c\circ\gamma(\vec{u})=\Sol_{k,r}^c\circ\gamma(\vec{v})$ if, and only if, $\gamma(\vec{u})=\gamma(\vec{v})$ and for each $\vec{i}\in[k]^{(2r)}$ there exist some $M\in \mathfrak{P}_{V^r}(\mathbb{F})$ such that for all $\sigma\in\im(\gamma)$
$$
\chi_{\sigma}M=M\xi_{\sigma}.
$$
In particular, $\gamma$ is $\Sol_{k,r}^c$-stable if, and only if, for all $\vec{u},\vec{v}\in V^k$ and $\vec{i}\in[k]^{(2r)}$
$$
\gamma(\vec{u})=\gamma(\vec{v})\implies \exists M\in \mathfrak{P}_{V^r}(\mathbb{F}),\chi_{\sigma}M=M\xi_{\sigma}\,\,\forall \sigma\in\im(\gamma).
$$
\end{proposition}

We now show some properties of the operator $\Sol_{k,r}^c$ in order to establish the link between the above and finite variable fragments of the extension of first-order logic with solvability quantifiers.
\begin{lemma}\label{sec3}
An $\IM_{k,r}^c$-stable partition is $\Sol_{k,r}^c$-stable.
\end{lemma}
\begin{proof}
Let $\mathbb{F}$ be the prime field of characteristic $c$. Suppose $\gamma$ is $\IM_{k,r}^c$-stable and let $\gamma(\vec{u})=\gamma(\vec{v})$. Then, for all $\vec{i}\in [k]^{(2r)}$, there is some $M\in GL_{V^r}(\mathbb{F})$ such that $M^{-1}\chi_{\sigma}M=\xi_{\sigma}$ for all $\sigma\in\im(\gamma)$. By equation \ref{sum}, $M\mathbb{J}_{V^r}=\mathbb{J}_{V^r}M$ and thus, both $M$ and $M^t$ have $\mathbf{1}_{V^r}$ as eigenvector with corresponding non-zero eigenvalue $\lambda\in\mathbb{F}$. Since, $\frac{1}{\lambda}M\in\mathfrak{P}_{V^r}(\mathbb{F})$, the result follows.
\end{proof}

\begin{lemma}
If $\gamma$ is $\Sol_{k,r}^c$-stable the following hold:
\begin{enumerate}
\item If $\gamma(\vec{u})=\gamma(\vec{v})$, and $\gamma$ is
  graph-like, then for all $\vec{i}\in[k]^{(2r)}$ and
  $\sigma\in\im(\gamma)$, $\{\vec{w}\in V^r \mid \gamma(\vec{u}\langle
  \vec{i},\vec{w}\cdot\vec{w}\rangle )=\sigma\}$ is non-empty if,
  and only if, $\{\vec{w}\in V^r \mid \gamma(\vec{v}\langle \vec{i},\vec{w}\cdot\vec{w}\rangle )=\sigma\}$ is non-empty.
\item If $c=0$, then $\gamma$ is $\C_{k,2r}$-stable.
\end{enumerate}
\end{lemma}
Note that if $c=0$, the second statement implies the first (refer to Section $4$ and $5$ of \cite{addv} for more details on properties of $\C_{k,r}$-stability).
\begin{proof}
Suppose $\{\vec{w}\in V^r \mid \gamma(\vec{u}\langle
\vec{i},(\vec{w}\cdot\vec{w})\rangle )=\sigma\}$ is non-empty. Since
$\gamma$ is graph-like, $\chi_{\sigma}$ must have all the non-zero
entries on the diagonal. Hence, for any
$M\in\mathfrak{P}_{V^r}(\mathbb{F})$, $\chi_{\sigma}M$ and,
consequently $M\xi_{\sigma}$ must be non-zero. As such, $\{\vec{w}\in
V^r \mid \gamma(\vec{v}\langle \vec{i},(\vec{w}\cdot\vec{w})\rangle )=\sigma\}$ is non-empty. The converse can be argued by symmetry, thus showing $(1)$.

For $(2)$, it suffices to show that: if $A,B\in\mathrm{Mat}_V(\mathbb{F})$ are $01$-matrices and $AM=MB$ for some $M\in\mathfrak{P}_V(\mathbb{F})$ then $A$ and $B$ have the same number of non-zero entries if $\mathrm{char}(\mathbb{F})=0$. Indeed, note that if $\alpha$ and $\beta$ are the number of non-zero entries of $A$ and $B$ respectively, then $\alpha\mathbb{J}_V=\mathbb{J}_VA\mathbb{J}_V$ and $\beta\mathbb{J}_V=\mathbb{J}_VB\mathbb{J}_V$. Since $M\mathbb{J}_V=\mathbb{J}_VM=\mathbb{J}_V$,
$$
\alpha\mathbb{J}_V=\mathbb{J}_VA\mathbb{J}_V=\mathbb{J}_VAM\mathbb{J}_V=\mathbb{J}_VMB\mathbb{J}_V=\mathbb{J}_VB\mathbb{J}_V=\beta\mathbb{J}_V,
$$
whence $\alpha=\beta$.
\end{proof}
In particular, if $r=1$, statement $(1)$ above implies that there is some $x\in V$ such that $\pr_{k-1}\gamma(\pr_{k-1}\vec{u}\langle i,x\rangle)=\sigma$ if, and only if, there is some $y\in V$, such that $\pr_{k-1}\gamma(\pr_{k-1}\vec{v}\langle i,y\rangle)=\sigma$. Furthermore, from $(2)$ and Lemma $5.7$ in \cite{addv}, $\pr_{k-1}\gamma$ is $\C_{k-1}$ stable.
\begin{corollary}\label{basiscase}
Let $\gamma=[\alpha_{k,\Gamma}]^{\Sol_{k,1}^c}$ for some graph $\Gamma$ on $V$. If $\pr_{k-1}\gamma(\vec{u})=\pr_{k-1}\gamma(\vec{v})$, there are no first-order formulae distinguishing $(\mathfrak{A}_\Gamma,\vec{z}\mapsto\vec{u})$ from $(\mathfrak{A}_\Gamma,\vec{z}\mapsto\vec{v})$. In addition, if $c=0$, there are no $\mathcal{C}_{k-1}$-formulae distinguishing  $(\mathfrak{A}_\Gamma,\vec{z}\mapsto\vec{u})$ from $(\mathfrak{A}_\Gamma,\vec{z}\mapsto\vec{v})$.
\end{corollary}
Similarly to the invertible map operators, we define the $k$-refinement operator $\Sol_k^c$ so that for $\gamma\in\mathcal{P}(V^k)$, the colour of $\vec{v}\in V^k$ in $\Sol_k^c\circ\gamma$ is given by a tuple whose entries are the colours of $\vec{v}$ in $\Sol_{k,r}^c\circ\gamma$, for all $r<k/2$; that is,
$$
\Sol_k^c\circ\gamma(\vec{v})=(\Sol_{k,1}^c\circ\gamma(\vec{v}),\Sol_{k,2}^c\circ\gamma(\vec{v}),\hdots,\Sol_{k,\lfloor k/2\rfloor}^c\circ\gamma(\vec{v})).
$$
The next statement will be the crux of our main results.
\begin{theorem}\label{fund}
For any prime number $p$, if $\gamma=[\alpha_{k,\Gamma}]^{\Sol_k^p}$ and $\gamma(\vec{u})=\gamma(\vec{v})$, there are no FOS$_{k-1}(p)$ formulae distinguishing $(\mathfrak{A}_\Gamma,\vec{z}\mapsto\vec{u})$ from $(\mathfrak{A}_\Gamma,\vec{z}\mapsto\vec{v})$.
\end{theorem}
\begin{proof}
We proceed by induction on the structure of FOS$(p)$ formulae
$\phi(\vec{z})$, where $\vec{z}$ is a $k$-tuple of pairwise distinct
variables. If $\phi(\vec{z})$ contains only atomic formulae, boolean
connectives and first order quantifiers, the statement holds by
Corollary \ref{basiscase}. Assume that for some $\phi(\vec{z})\in
\text{FOS}_k(p)$,
$\mathfrak{A}_\Gamma\models\phi(\vec{u})\iff\mathfrak{A}_\Gamma\models\phi(\vec{v})$,
and suppose $(\mathfrak{A}_\Gamma,\vec{z}\mapsto\vec{v})\models
\mathrm{slv}_p[\vec{x}\vec{y}.\phi(\vec{z}\langle
\vec{i},\vec{x}\cdot\vec{y}\rangle)]$, where $\vec{i}\in[k]^{(2r)}$
and $\vec{x},\vec{y}$ are $r$-tuples of distinct variables (distinct
from variables in the tuple $\vec{z}$). Let $S^{\vec{v}}$ be the adjacency matrix of the relation $\{(\vec{a},\vec{b})\in V^r\times V^r\mid (\mathfrak{A}_\Gamma,\vec{z}\mapsto\vec{v})\models\phi(\vec{v}\langle\vec{i},\vec{a}\cdot\vec{b}\rangle)\}$, and similarly define $S^{\vec{u}}$. Then, there is some $\vec{a}\in \mathbb{Z}_p^{V^r}$ such that $S^{\vec{v}}\vec{a}=\mathbf{1}_{V^r}$. By the induction hypothesis, $S^{\vec{u}}=\sum_{\sigma\in I}\chi_{\sigma}$ and $S^{\vec{v}}=\sum_{\sigma\in I}\xi_{\sigma}$ for some $I\subseteq\im(\gamma)$. Since there exists $M\in\mathfrak{P}_{V^r}(\mathbb{Z}_p)$ such that $\chi_{\sigma}M=M\xi_{\sigma}$ for all $\sigma\in\im(\gamma)$, we have
$$
S^{\vec{v}}\vec{a}=\mathbf{1}_{V^r}\implies MS^{\vec{v}}\vec{a}=\mathbf{1}_{V^r}\implies S^{\vec{u}}(M\vec{a})=\mathbf{1}_{V^r},
$$
from which it follows that $(\mathfrak{A}_\Gamma,\vec{z}\mapsto\vec{u})\models\mathrm{slv}_p[\vec{x}\vec{y}.\phi(\vec{z}\langle\vec{i},\vec{x}\cdot\vec{y}\rangle)]$. Using a symmetric argument, we conclude that $\mathrm{slv}_p[\vec{x}\vec{y}.\phi(\vec{z}\langle\vec{i},\vec{x}\cdot\vec{y}\rangle)]$ does not distinguish $(\mathfrak{A}_\Gamma,\vec{z}\mapsto\vec{u})$ from $(\mathfrak{A}_\Gamma,\vec{z}\mapsto\vec{v})$ and the result follows by induction.
\end{proof}

\subsection{Polynomial, monomial and Nullstellensatz calculi}\label{calculi}
The idea behind the proof systems we are concerned with is that of
encoding Boolean formulae by multivariate polynomials over a field and
determining, via inference rules, if these do not have a common
root. In such cases, we conclude that the formulae are
inconsistent. The \emph{polynomial calculus} (PC), \emph{monomial
  calculus} (MC) and \emph{Nullstellensatz calculus} (NC) are proof
systems of this kind that have been studied in the literature.  The PC inference rules for a set of axioms $A\subset \mathbb{F}[x_1,...,x_n]$ are as follows:
\begin{enumerate}
\item $\frac{}{f}$ for all $f\in A$.
\item \emph{Multiplication rule}: $\frac{f}{xf}$ for all derived polynomials $f$, and variables $x\in\{x_1,...,x_n\}$.
\item \emph{Linearity rule}: $\frac{f,g}{\lambda f+\mu g}$ for all derived polynomials $f,g$ and $\lambda,\mu\in\mathbb{F}$.
\end{enumerate}
The inference rules for MC are obtained by restricting $f$ in the
multiplication rule to be an axiom times a monomial or a monomial. By
further restricting $f$ to be an axiom times a monomial, one obtains
the NC inference rules. The \emph{degree} of a PC (or MC or NC)
derivation is the maximum degree of all polynomials involved in the
derivation, and a PC (or MC or NC, respectively) \emph{refutation} of
$A$ is a derivation of $1$ from $A$ using PC (or MC or NC, respectively) rules. We are really interested in roots where the variables are assigned $01$-values, so as to encode the Boolean framework. To enforce this, we tacitly assume that the set of axioms always includes the polynomials $x^2-x$ for all $x\in\{x_1,\hdots,x_n\}$. As such, we may restrict our focus to multilinear polynomials exclusively: indeed, if, for distinct variables $x$ and $y$, we may derive $x^2y$, we can also derive $x^2y-y(x^2-x)=xy$.

It is easy to check that all three proof systems are sound and, by Hilbert's weak Nullstellensatz Theorem\footnote{This is a fundamental result in algebra which states that the polynomials of an ideal in $\mathbb{F}[x_1,...x_n]$ have a common root in the algebraic closure of $\mathbb{F}$ if, and only if, $1$ is not in the ideal.}, complete. That is, there is a refutation in each of the proof systems if, and only if, the polynomials in the axioms do not have a common root. In general, the degree of refutations can be unbounded, which suggests that such a quantity is a reasonable complexity parameter to study. We denote by PC$_k$ the proof system using the same inference rules as PC with the added constraint that all derivations must have degree at most $k$, and we use the same convention for MC$_k$ and NC$_k$. Though these proof systems with bounded degree are not complete, their refutations can be decided in polynomial time in the number of variables.

For a graph $\Gamma$ on $V$ we define $\mathrm{Ax}(\Gamma)\subseteq\mathbb{F}[\{x_{uv}|u,v\in V\}]$ to be the set of axioms containing the following polynomials:
\begin{itemize}
\item[\textbf{A1.}] $\sum_{u\in V}x_{uv}-1\;\,\textrm{ for all $v\in V$.}$
\item[\textbf{A2.}] $\sum_{u\in V}x_{vu}-1\;\,\textrm{ for all $v\in V$.}$
\item[\textbf{A3.}] $x_{uv}x_{u'v'}\;\,\textrm{ if the map $u\mapsto u',v\mapsto v'$ is not a local isomorphism in $\Gamma$.}$
\item[\textbf{A4.}] $x_{uv}^2-x_{uv}\;\,\textrm{for all $u,v\in V$.}$
\end{itemize}
For $\vec{u},\vec{v}\in V^k$, we further define $\mathrm{Ax}(\Gamma_{\vec{u}\rightarrow \vec{v}})$ to contain the above plus the following:
\begin{itemize}
\item[\textbf{A5.}] $x_{v_iu_i}-1\;\,\textrm{for all $i\in [k]$.}$
\end{itemize}

Note that any of these proof systems has a refutation of $\mathrm{Ax}(\Gamma_{\vec{u}\rightarrow\vec{v}})$ if, and only if, $\vec{u}$ and $\vec{v}$ are in different equivalence classes of the $k$-orbit partition for $\Gamma$
. Furthermore, any of the proof systems has a degree $d$ refutation of $\mathrm{Ax}(\Gamma_{\vec{u}\rightarrow\vec{v}})$ if, and only if, it allows a degree $d$ derivation from $\mathrm{Ax}(\Gamma)$ of the monomial $\prod_{i\in J}x_{v_iu_i}$ for some subset $J\subseteq [k]$ of size at most $d$. Since all coefficients of the elements of $\mathrm{Ax}(\Gamma_{\vec{u}\rightarrow\vec{v}})$ are $0$ and $\pm1$, refutations can be decided by restricting the use of the linearity rule with $\lambda,\mu$ belonging to the prime subfield of $\mathbb{F}$. When considering these axioms, we may then assume, without loss of generality, that $\mathbb{F}$ is a prime field.

Let $\equiv_{\PC_k}^c$ be the relation on $V^k$, where $\vec{u}\equiv_{\PC_k}^c\vec{v}$ if there is no degree $k$ PC refutation of $\mathrm{Ax}(\Gamma_{\vec{u}\rightarrow\vec{v}})$ over the prime field of characteristic $c$, and similarly define $\equiv_{\MC_k}^c$ and $\equiv_{\NC_k}^c$.
\begin{lemma}
	$\equiv_{\PC_k}^c$, $\equiv_{\MC_k}^c$ and $\equiv_{\NC_k}^c$ are equivalence relations on $V^k$.
\end{lemma}
We prove this statement only for $\equiv_{\PC_k}^c$, since a similar argument applies in the other two cases. As remarked above, we may restrict our focus to multilinear polynomials exclusively and denote the multilinear monomial $x_{a_1}x_{a_2}\hdots x_{a_m}$ by $X_A$ where $A=\{a_1,a_2,\hdots a_m\}$.
\begin{proof}
	Clearly, $\vec{u}\equiv_{\PC_k}^c\vec{u}$. Indeed, the polynomials in the ideal generated by $\mathrm{Ax}(\Gamma_{\vec{u}\rightarrow\vec{u}})$ have the common root  $x_{rs}=\delta_{rs}$, where $\delta_{rs}$ is the Kronecker delta. Hence, the ideal generated by $\mathrm{Ax}(\Gamma_{\vec{u}\rightarrow\vec{u}})$ is non-trivial. The set $\mathrm{Ax}(\Gamma)$ is invariant under the transformation $x_{rs}\rightarrow x_{sr}$ for all $r,s\in V$. Thus, $\vec{u}\equiv_{\PC_k}^c\vec{v}\implies\vec{v}\equiv_{\PC_k}^c\vec{u}$.
	
	Suppose $\vec{u}\equiv_{\PC_k}^c\vec{v}$ and $\vec{v}\equiv_{\PC_k}^c\vec{w}$. Let $\pi$ be the map $v_i\rightarrow w_i$ for $i\in[k]$. Note that $\pi$ is well defined, for if $v_i=v_j$ for some $i\neq j$ and $\pi(v_i)\neq \pi(v_j)$, then $x_{v_iw_i}x_{v_jw_j}\in\mathrm{Ax}(\Gamma)$ - a contradiction. For $A\subseteq V^2$, define $A^\pi$ as follows:
	$$
	A^\pi=\begin{cases}
	\{(r,\pi(s))\mid (r,s)\in A\}\;\;\textrm{ if for all $(r,s)\in A$ there is some $j$ such that $s=v_j$;}\\
	\{(r,\pi^{-1}(s))\mid (r,s\in A)\}\;\,\textrm{ if for all
          $(r,s)\in A$ there is some $j$ such that $s=w_j$; and}\\
	A\;\;\textrm{otherwise.}
	\end{cases}
	$$
	For a multilinear polynomial $f=\sum a_AX_A$, let $p^\pi=\sum
        a_AX_{A^\pi}$. We show by induction on the $\PC$ inference
        rules that there is a $\PC_k$ derivation of $p$ if, and only
        if, there is $\PC_k$ derivation of $f^\pi$. Indeed, if $f$ is
        in $\mathrm{Ax}(\Gamma)$, then so is $f^\pi$. Suppose there is a $\PC_k$ derivation of $f,g,f^\pi$ and $g^\pi$. Then there is a $\PC_k$ derivation of $(\lambda f+\mu g)^\pi=\lambda f^\pi+\mu g^\pi$. Finally, suppose the degree of $f$ is less than $k$, and suppose, without loss of generality that $f=X_A$ for some $A\subseteq V^2$. Then, there is a $\PC_k$ derivation of $X_{A\cup\{(r,s)\}}$ for any $r,s\in V$. By checking case by case, it follows that there is a $\PC_k$ derivation of $(X_{A\cup\{(r,s)\}})^\pi$. Since $(A^\pi)^\pi=A$ and hence, $(f^\pi)^\pi=f$, there is a $\PC_k$ derivation of $f$ if, and only if, there is a $\PC_k$ derivation of $f^\pi$. In particular, since there is no $\PC_k$ derivation of $X_{\{(v_i,w_i)\mid i\in[k]\}}$, there is no derivation of $(X_{\{(u_i,v_i)\mid i\in[k]\}})^\pi=X_{\{(u_i,w_i)\mid i\in[k]\}}$. Whence, $\vec{u}\equiv_{\PC_k}^c\vec{w}$.
\end{proof}

It is easy to see that the relation $\equiv_{\PC_k}^c$ refines $\equiv_{\MC_k}^c$ which, in turn, is a refinement of $\equiv_{\NC_k}^c$, since any $\NC_k$ refutation is a $\MC_k$ refutation which is also a $\PC_k$ refutation. More precisely:

\begin{lemma}
	For any graph on $V$ and $\vec{u},\vec{v}\in V^k$, $\vec{u}\equiv_{\PC_k}^c\vec{v}\implies\vec{u}\equiv_{\MC_k}^c\vec{v}\implies\vec{u}\equiv_{\NC_k}^c\vec{v}$.
\end{lemma}

For $c=0$, Grohe et.\ al.\ have characterized these equivalence relations in terms of counting logics. More precisely:
\begin{displayquote}
Let $\Gamma$ be a graph on $V$ and $\vec{u},\vec{v}\in V^k$. Then $\vec{u}\equiv_{\MC_k}^0\vec{v}$ if, and only if, no $\mathcal{C}_k$ formula distinguishes $(\mathfrak{A}_\Gamma,\vec{z}\mapsto\vec{u})$ from $(\mathfrak{A}_\Gamma,\vec{z}\mapsto\vec{v})$ (Theorem $4.4$ in \cite{berkh}). Furthermore, if $\vec{u}\not\equiv_{\PC_k}^0\vec{v}$, there a  $k'=O(k)$, such that some $\mathcal{C}_{k'}$ formula distinguishes $(\mathfrak{A}_\Gamma,\vec{z}\mapsto\vec{u})$ from $(\mathfrak{A}_\Gamma,\vec{z}\mapsto\vec{v})$ (Theorem 6.6 in \cite{g2p2}).
\end{displayquote}

In this paper, we aim at a similar characterization for the equivalence relations induced by the proof systems in fields of positive characteristic.  The corresponding logic we consider are the finite variable fragments of the extension of first-order logic with solvability quantifiers.  This allows us to relate the distinguishing power of these relations to that of the equivalence relations defined by the invertible map tests.

\section{Definability of monomial calculus refutations over finite fields}\label{sec:monomial}

At the core of the proof of Theorem \ref{thm:main1} is the definability of monomial calculus refutations in FPS$(p)$. More precisely, the main objective of this section is to prove the following statement.
\begin{lemma}\label{4.1}
	Let $\mathfrak{A}$ be a structural encoding of a finite set of polynomials $P$ of degree at most $d$, over a finite field $\mathbb{F}$ of positive characteristic $p$. For any $d,k\in\nats$ there is a $\mathrm{FPS}(p)$ formula $\phi_{d,k}$ such that $\mathfrak{A}\models\phi_{d,k}$ if, and only if, there is an $\MC_k$ refutation of $P$ over $\mathbb{F}$. \footnote{Note that it is possible to drop the dependence on $d$, but this would involve a more cumbersome definition of $\mathfrak{A}$ and subsequent interpretations. Since elements of $\mathrm{Ax}(\Gamma)$ have degree at most $2$, this statement is sufficient for our purpose.}
\end{lemma}
For the sake of argument, we assume that $\mathbb{F}=\mathbb{Z}_p$. We first recall how to express the solvability of linear equations with coefficients other than $0$ and $1$ and explain the meaning of \emph{structural encoding} of a set of polynomials.
\subsection{Defining solvability of linear equations over finite fields}\label{solsec}
For each prime number $p$, let $\mathtt{LIN}_p$ be a relational vocabulary with the following symbols:
\begin{enumerate}
	\item A binary relational symbol $\texttt{A}_q$ for each $q\in\mathbb{Z}_p$.
	\item A unary relational symbol $\texttt{b}_q$ for each $q\in\mathbb{Z}_p$.
\end{enumerate}
Let $A\in\mathrm{Mat}_{E\times V}(\mathbb{Z}_p)$ and $\vec{b}\in\mathbb{Z}_p^E$. A $\mathtt{LIN}_p$-structure $\mathfrak{A}$ with universe $V\cup E$ ($V$ for variables and $E$ for equations) encodes the system of linear equations $A\vec{x}=\vec{b}$ if, for all $e\in E,v\in V$, $\mathfrak{A}\models\texttt{A}_q(e,v)$ if $A_{ev}=q$ and $\mathfrak{A}\models\texttt{b}_q(e)$ if $b_e=q$.

Recall Lemma 4.1 in \cite{defin}.
\begin{lemma}\label{dhk}
	There is a quantifier free interpretation $\mathcal{I}$ of $\mathtt{LIN}_p$ into $\mathtt{LIN}_p$ such that if $\mathfrak{A}$ encodes the system of linear equations $A\vec{x}=\vec{b}$, then:
	\begin{enumerate}
		\item $\mathcal{I}(\mathfrak{A})$ encodes a system of linear equations $A'\vec{y}=\mathbf{1}$, where $\mathbf{1}$ is the all $1$s vector of appropriate length and $A'$ is a $01$-matrix.
		\item $A'\vec{y}=\mathbf{1}$ has a solution if, and only if, $A\vec{x}=\vec{b}$ has a solution.
	\end{enumerate}
\end{lemma}
Thus, $\mathcal{I}(\mathfrak{A})\models\mathrm{slv}_p(xy.\mathtt{A}_1(x,y))$ if, and only if, $A\vec{x}=\vec{b}$ has a solution and hence, from the Interpretation Lemma, there is a FOS$(p)$ formula $\Phi$ such that $\mathfrak{A}\models\Phi$ if, and only if, the system of linear equations encoded by $\mathfrak{A}$ has a solution.

\subsection{Proof of Lemma \ref{4.1}}
Monomial calculus refutations of bounded degree can be understood as the following procedure, where the input is a finite set of multilinear polynomials $P$ from the ring $\mathbb{F}[x_1,\hdots,x_r]$, and the output is $\texttt{REFUTE}$ or $\texttt{NOREFUTE}$. We denote the multilinear monomial $x_{a_1}x_{a_2}\hdots x_{a_r}$ by $X_A$ where $A=\{a_1,a_2\hdots,a_r\}$, so that for a polynomial $f$, $X_Af=x_{a_1}x_{a_2}\hdots x_{a_r}f$. In this form, $X_\emptyset$ denotes $1$.

\begin{enumerate}
	\item[INPUT.] $P\subset \mathbb{F}[x_1,\hdots,x_r]$
	\item[OUTPUT.] \texttt{REFUTE} or \texttt{NOREFUTE}.
	\item Initialize $\mathcal{S}=\{X_Af\mid f\in P, \mathrm{deg}(X_Af)\leq k\}$.
	\item \texttt{while} $\mathrm{span}_\mathbb{F}\mathcal{S}$ has changed since last round or $1\notin \mathrm{span}_\mathbb{F}\mathcal{S}$ \texttt{do lines 3-4}
	\item $\mathcal{M}\leftarrow \{A\subseteq[r]\mid |A|< k, X_A\in\mathrm{span}_\mathbb{F}\mathcal{S}\}.$
	\item $\mathcal{S}\leftarrow\mathcal{S}\cup \{X_B\mid |B|\leq k, \exists A\in \mathcal{M}, A\subseteq B\}.$
	\item \texttt{If} $1\in \mathrm{span}_\mathbb{F}\mathcal{S}$ \texttt{output REFUTE}
	\item \texttt{else output NOREFUTE}.
\end{enumerate}
Put otherwise, the procedure computes the inflationary fixed point of the set of monomials within the $\mathbb{F}$-span of $\mathcal{S}$ and checks whether it contains the constant $1$ or its dimension has increased. This does not require to store in memory the set $\mathrm{span}_\mathbb{F}\mathcal{S}$ (whose size is exponential in the input), but can be done by checking the solvability of linear equations. The number of iterations of the \texttt{while} loop is at most the number of multilinear monomials of degree at most $k$, thus ensuring that the procedure runs, overall, in polynomial time. Crucially, at each iteration of the \texttt{while} loop, the $\mathbb{F}$-span of $\mathcal{S}$ has a \emph{canonical} generating set. As illustrated below, this property is really what allows definability of the above procedure in FPS$(p)$.

Recall that we are assuming that $\mathbb{F}=\mathbb{Z}_p$. By viewing polynomials over $\mathbb{Z}_p$ as vectors in the standard basis given by monomials, we encode $P$ as a structure $\mathfrak{A}$ with universe $V$ over the vocabulary $\mathtt{POLY}_p=(\mathtt{Var}, \mathtt{U},  \mathtt{C}_0,\mathtt{C}_1,\hdots,\mathtt{C}_{p-1})$, where:
\begin{enumerate}
	\item $V=P\cup \{x_i\mid i\in[r]\}\cup\{1\}$.
	\item $\mathtt{Var}$ and $\mathtt{U}$ are unary relational symbol with $\mathfrak{A}\models \mathtt{Var}(v)$ if, and only if $v\in\{x_i\mid i\in[r]\}$, and $\mathfrak{A}\models\mathtt{U}(v)$ if, and only if, $v=1$.
	\item $\mathtt{C}_q$ for $q\in\mathbb{Z}_p$ are $(d+1)$-ary relations, where $d$ is the maximal degree of the polynomials in $P$. These encode $P$ in matrix form; that is, $\mathfrak{A}\models \mathtt{C}_q(u,v_1,\hdots, v_d)$ if, and only if, $u\in P$, each $v_i$ is equal to a variable or $1$, and the coefficient of the monomial $v_1v_2\hdots v_d$ in $u$ is equal to $q$.
\end{enumerate}

Note that there is a formula $\delta_m$ identifying the polynomials in $P$ of degree $m$. Indeed, let $\kappa_{m,t}(v_1,\hdots,v_t)$ for each $t\geq m$ be a formula satisfied by $\mathfrak{A}$ if, and only if, $v_1,\hdots,v_t$ assume exactly $m$ distinct values from $\{x_i\mid i\in[r]\}$, and define
$$
\hat\delta_m(u)=\bigvee_{q\in\mathbb{Z}_p^*}\exists z_1\hdots z_d.\mathtt{C}_q(u,z_1,\hdots,z_d)\land \kappa_{m,d}(z_1,\hdots,z_d)
$$
Then $\delta_m(u)=\hat\delta_m(u)\land\lnot\hat\delta_{m+1}(u)$ is satisfied by $\mathfrak{A}$ if, and only if, $u\in P$ has degree $m$.

Next, we use a $k$-ary interpretation to obtain a structure whose universe is the set of multilinear polynomials of degree at most $k$. Formally, let $\mathtt{POLY}_p^*=(\mathtt{A}_0,\mathtt{A}_1,\hdots,\mathtt{A}_{p-1},\mathtt{U}^*, \mathtt{Mon},\mathtt{Sub})$ be a vocabulary where $\mathtt{Mon}$ and $\mathtt{U}^*$ are unary relational symbols, $\mathtt{Sub}$ is a binary relational symbol, and $\mathtt{A}_q$ are as in the vocabulary $\mathtt{LIN}_p$. One can then define an interpretation $\mathcal{J}$ from $\mathtt{POLY}_p$ into $\mathtt{POLY}^*_p$ such that:
\begin{enumerate}
	\item The universe of $\mathcal{J}(\mathfrak{A})$ are elements $\vec{v}\in V^k$ such that
	\begin{equation}\label{form}
		\mathfrak{A}\models\bigwedge_{i\in[k]}\mathtt{U}(v_i)\lor\bigwedge_{i\in[k]}\mathtt{Mon}(v_i) \lor \bigvee_{m\in[k]}\delta_m(v_1)\land \kappa_{k-m,k}(v_2,\hdots,v_k),
	\end{equation}
	and the equivalence relation $\equiv_\mathcal{J}$ on $V^k$ defined by the equivalence formula of $\mathcal{J}$ is as follows: $\vec{u}\equiv_\mathcal{J}\vec{v}$ if, and only if, the set of entries of $\vec{u}$ is the same as that of $\vec{v}$. Put otherwise, formula (\ref{form}) ensures that the elements of the universe of $\mathcal{J}(\mathfrak{A})$ are tuples $\vec{v}$ where either all $v_i$ are equal to $1$, all $v_i$ are variables, or $v_1$ is a polynomial in $P$ and $v_2,\hdots,v_k$ are either all equal to $1$ or are variables such that $|\{v_i\mid 2\leq i\leq k\}|\leq k-\mathrm{deg}(v_1)$. The relation $\equiv_\mathcal{J}$ partitions the universe into equivalence classes uniquely determined by the set of entries of each tuple. If the entries of $\vec{v}$ are all equal to $1$, we indicate its class by $X_\emptyset$, and similarly, if $v_i=x_{a_i}$ for all $i\in[k]$, we indicate its class by $X_A$ where $A=\{a_i\mid i\in[k]\}$. If $v_1=f$ for some $f\in P$ and $v_2,\hdots,v_k$ are all equal to $1$, we indicate the class of $\vec{v}$ by $f$ and if $v_i=a_i$ for $2\leq i\leq k$, we indicate its class by the pair $(X_A,f)$ where $A=\{a_i\mid 2\leq i\leq k\}$.
	\item $\mathfrak{A}\models\mathtt{U}(\vec{v})$ and $\mathfrak{A}\models\mathtt{Mon}(\vec{v})$  if, and only if, the equivalence class of $\vec{v}$ is $X_\emptyset$ and $X_A$ with $|A|\geq 1$ respectively.
	\item $\mathfrak{A}\models\mathtt{A}_q(\vec{u},\vec{v})$ if, and only if, the equivalence class of $\vec{u}$ and $\vec{v}$ are $X_A$ and $(X_B,f)$ for some $f\in P$ respectively, and the coefficient of $X_A$ in $X_Bf$ equals $q$. 
	\item $\mathfrak{A}\models \mathtt{Sub}(\vec{u},\vec{v})$ if, and only if, the equivalence classes of $\vec{u}$ and $\vec{v}$ are $X_A$ and $X_B$ respectively and $A\subseteq B$.
\end{enumerate}
We leave the details to the reader to check that the above relations on $V^k$ can be defined by first-order formulae.

Thus far, we have shown that the \texttt{INPUT} of the monomial calculus procedure and the initial set $\mathcal{S}$ in line $1$ can be defined in first-order logic. The last thing required for proving Lemma \ref{4.1} is showing the definability of monomials in the set $\mathcal{S}$ \emph{after} each iteration of the \texttt{while} loop, which is where the solvability quatifier and fixed-point operator come into play. In what follows, set $\mathcal{T}=\{X_Af\mid f\in P, A\subseteq [r], \textrm{deg}(X_Af)\leq k\}$.
\begin{proof}[Proof of Lemma \ref{4.1}] Let $\psi(z)$ be some FPS$(p)$ formula over the vocabulary \texttt{POLY}*$_p$. There is an interpretation $\mathcal{K}(t)$ of $\mathtt{POLY}_p^*$ into $\mathtt{LIN}_p$ such that $\mathcal{K}(\mathcal{J}(\mathfrak{A}), t\mapsto X_A)$ encodes the system of linear equations determining whether $X_A$ is in the $\mathbb{Z}_p$-span of the set $$\mathcal{T}_\psi=\mathcal{T}\cup\{X_B\mid(\mathfrak{A},z\mapsto X_B)\models\psi(z)\land\mathtt{Mon}(z)\}.$$
By Lemma \ref{dhk}, there is a FPS$(p)$ formula $\theta_\psi(z)$, depending on $\psi$, such that $(\mathcal{J}(\mathfrak{A}), z\mapsto X_A)\models \theta_\psi(z)$ if, and only if, the monomial $X_A$ is in the $\mathbb{Z}_p$-span of $\mathcal{T}_\psi$.
	
In particular, replacing $\psi$ with some unary relational variable $Z$, then $$(\mathcal{J}(\mathfrak{A}),z\mapsto X_A)\models \exists y.\theta_Z(y)\land\mathtt{Sub}(y,z)$$
holds if, and only if, there is some $B\subseteq A$ such that $X_B$ is in the $\mathbb{Z}_p$-span of $\mathcal{T}_Z$. Whence,
\begin{equation}\label{eq:sol}
(\mathcal{J}(\mathfrak{A}),z'\mapsto X_\emptyset)\models \mathrm{ifp}_{Z,z}(\exists y.\theta_Z(y)\land\mathtt{Sub}(y,z))(z')
\end{equation}
if, and only if, there is an $\MC_k$ refutation of $P$ over $\mathbb{Z}_p$. By applying the Interpretation Lemma to the above formula and the interpretation $\mathcal{J}$, the desired result follows.
\end{proof}
Note that in formula \ref{eq:sol}, the solvability quantifier is included in the formula $\theta_Z(y)$.
\begin{corollary}\label{cory}
For any $k\in \nats$, there is some $k'$ such that for any graph $\Gamma$ on $V$ and $\vec{u},\vec{v}\in V^{k}$, if there is a $\MC_k$ refutation of $\mathrm{Ax}(\Gamma_{\vec{u}\rightarrow \vec{v}})$ over $\mathbb{Z}_p$, then there is some FOS$_{k'}(p)$ formula $\phi(\vec{z})$ distinguishing $(\mathfrak{A}_\Gamma, \vec{z}\mapsto\vec{u})$ from $(\mathfrak{A}_\Gamma,\vec{z}\mapsto\vec{v})$.
\end{corollary}
\begin{proof}
By Lemma \ref{4.1} the equivalence classes of $\equiv_{\MC_k}^p$ are definable in FPS$(p)$ and hence, by the embedding of FPS$(p)$ in infinitary FOS$(p)$, the statement follows.
\end{proof}
Theorem \ref{thm:main1} readily follows.
\begin{proof}[Proof of Theorem \ref{thm:main1}]
Statement $(1)$ follows directly from the relation between the distinguishing power of counting logics, polynomial calculus and the invertible map tests highlighted in Section $3$.

Let $\vec{u},\vec{v}\in V^k$ and suppose $\mathrm{Ax}(\Gamma_{\vec{v}\rightarrow \vec{u}})$ has a $\MC_k$ refutation over $\mathbb{Z}_p$. By Corollary \ref{cory} there is a FOS$_k(p)$ formula $\phi(\vec{z})$ distinguishing the structures $(\mathfrak{A}_\Gamma,\vec{z}\mapsto\vec{u})$ and $(\mathfrak{A}_\Gamma,\vec{z}\mapsto\vec{v})$. Then, by Theorem \ref{fund}, $\vec{u},\vec{v}$ are in different equivalence classes of $\pr_k[\alpha_{k',\Gamma}]^{\Sol_{k'}^p}$. Since $[\alpha_{k',\Gamma}]^{\IM_{k'}^p}\succeq[\alpha_{k',\Gamma}]^{\Sol_{k'}^p}$ (by Lemma \ref{sec3}), statement $(2)$ is implied.
\end{proof}

\section{Nullstellensatz refutations and $\Sol_k^p$-stability}\label{sec:nullstellensatz}

The focus of this section is the proof of the following statement.
\begin{lemma}\label{secondone}
	Let $\Gamma$ be a graph on $V$ and let $\gamma\in\mathcal{P}(V^k)$. If for all $\vec{u},\vec{v}\in V^k$, $\gamma(\vec{u})=\gamma(\vec{v})$ if, and only if, $\vec{u}\equiv_{\NC_k}^p\vec{v}$, then $\gamma$ is $\Sol_k^p$-stable.
\end{lemma}
Theorems \ref{thm:main2} and \ref{thm:main3} are its direct consequences.

\subsection{Proof of Lemma \ref{secondone}}
In what follows, $\gamma\in\mathcal{P}(V^k)$ satisfies the assumptions of Lemma \ref{secondone}, $\vec{u},\vec{v}\in V^k$ and $\vec{\chi}$ and $\vec{\xi}$ are the $(\vec{i},\vec{u})$ and $(\vec{i},\vec{v})$-character vectors of $\gamma$ respectively, where we may assume without loss of generality that $\vec{i}=(k,k-1,\hdots,k-2r+1)\in[k]^{(2r)}$, for some $r$ with $2r<k$. For $\vec{w},\vec{z}\in V^l$, denote by $X_{\vec{w}\vec{z}}$ the monomial $x_{w_1z_1}x_{w_2z_2}\hdots x_{w_lz_l}$ (which need not be multilinear).

Recall that $\gamma$ is $\Sol_{k,r}^p$-stable, if, and only if, for every $\vec{u},\vec{v}$ there is some matrix $T\in\mathfrak{P}_{V^r}(\mathbb{Z}_p)$ such that $\chi_\sigma T=T\xi_\sigma$, for all $\sigma\in\im (\gamma)$.   That is, the following system of equations is solvable in the variables $T_{\vec{w}\vec{z}}$:
$$
\sum_{\vec{a}\in V^r}(\chi_{\sigma})_{\vec{w}\vec{a}}T_{\vec{a}\vec{z}}-\sum_{\vec{a}\in V^r} T_{\vec{w}\vec{a}}(\xi_\sigma)_{\vec{a}\vec{z}}=0 \;\; \textrm{for  $\sigma\in\im(\gamma)$, $\vec{w},\vec{z}\in V^r$}
$$
$$
\sum_{\vec{a}\in V^r}T_{\vec{w}\vec{a}}-1=0\;\;\textrm{and}\;\;\sum_{\vec{a}\in V^r}T_{\vec{a}\vec{z}}-1=0\;\;\textrm{for $\vec{w},\vec{z}\in V^r$},
$$
where $(\chi_\sigma)_{\vec{w}\vec{a}}$ is equal to $1$ if $\gamma(\vec{u}\langle \vec{i},\vec{w}\cdot\vec{a}\rangle)=\sigma$ and $0$ otherwise (similarly for $\xi_\sigma$).
 
We show that there is a $\NC_k$ derivation from $\mathrm{Ax}(\Gamma_{\vec{v}\rightarrow \vec{u}})$ over $\mathbb{Z}_p$ of the following multilinear polynomials (Lemma \ref{thirdone}):
\begin{equation}\label{der1}
	\sum_{\vec{a}\mid \gamma(\vec{u}\langle\vec{i},\vec{w}\cdot\vec{a}\rangle)=\sigma}X_{\vec{a}\vec{z}}-\sum_{\vec{a}\mid \gamma(\vec{v}\langle\vec{i},\vec{a}\cdot\vec{z}\rangle)=\sigma}X_{\vec{w}\vec{a}}\;\; \textrm{for  $\sigma\in\im(\gamma)$, $\vec{w},\vec{z}\in V^r$}
\end{equation} 
\begin{equation}\label{der2}
	\sum_{\vec{a}\in V^r}X_{\vec{w}\vec{a}}-1\;\;\textrm{and}\;\;\sum_{\vec{a}\in V^r}X_{\vec{a}\vec{z}}-1\;\;\textrm{for $\vec{w},\vec{z}\in V^r$}.
\end{equation}
We may view each monomial (apart from the constant term) as a distinct linear variable, so that all of the above are linear polynomials. Since $\gamma(\vec{u})=\gamma(\vec{v})$, there is no $\NC_k$ refutation of $\mathrm{Ax}(\Gamma_{\vec{v}\rightarrow \vec{u}})$ and hence, no linear combination of \ref{der1} and \ref{der2} gives the constant polynomial $1$. It follows that if viewed as linear polynomials, \ref{der1} and \ref{der2} have a common root, thus showing Lemma \ref{secondone}.

\begin{lemma}\label{ll0}
	If $\gamma(\vec{u}\langle \vec{i},\vec{w}\cdot\vec{z}\rangle)\neq\gamma(\vec{v}\langle \vec{i},\vec{w}'\cdot\vec{z'}\rangle)$, then there is a $\NC_k$ derivation over $\mathbb{Z}_p$ from $\mathrm{Ax}(\Gamma_{\vec{v}\rightarrow \vec{u}})$ of $X_{\vec{w}\vec{w}'}X_{\vec{z}\vec{z}'}$.
\end{lemma}
\begin{proof}
	Set $Y=X_{\vec{w}\vec{w}'}X_{\vec{z}\vec{z}'}$ and let $X_{\vec{u}'\vec{v}'}$ be the degree $k-2r$ monomial where $u'_j=u_j$ and $v'_j=v_j$ for $j\in[k-2r]$. Then, there is a $\NC_k$ derivation of $Y(X_{\vec{u}'\vec{v}'}-1)$, for indeed
	$$
	Y(X_{\vec{u}'\vec{v}'}-1)=Y(x_{u_1v_1}-1)+Yx_{u_1v_1}(x_{u_2v_2}-1)+\hdots +Y(x_{u_1v_1}\hdots x_{u_{k-2r-1}v_{k-2r-1}})(x_{u_{k-2r}v_{k-2r}}-1).
	$$
	By assumption, $\gamma(\vec{u}\langle \vec{i},\vec{w}\cdot\vec{z}\rangle)\neq\gamma(\vec{v}\langle \vec{i},\vec{w}'\cdot\vec{z'}\rangle)$ and hence, there is an $\NC_k$ derivation of $YX_{\vec{u}'\vec{v}'}$. Subtracting the latter from $Y(X_{\vec{u}'\vec{v}'}-1)$ yields the desired statement.
\end{proof}
\begin{lemma}\label{ll1}
	For any $\vec{w},\vec{z}\in V^r$ and $\vec{s}\in V^t$ there is a $\NC_{t+r}$ derivation over $\mathbb{Z}_p$ from $\mathrm{Ax}(\Gamma_{\vec{v}\rightarrow \vec{u}})$ of
	$$
	X_{\vec{w}\vec{z}}\big(\sum_{\vec{a}\in V^t}X_{\vec{s}\vec{a}}-1\big)\;\;\textrm{and}\;\;X_{\vec{w}\vec{z}}\big(\sum_{\vec{a}\in V^t}X_{\vec{a}\vec{s}}-1\big).
	$$
\end{lemma}
\begin{proof}
	We proceed by induction on $t$. For $t=1$, $X_{\vec{w}\vec{z}}\big(\sum_{a\in V}x_{sa}-1\big)$ is the product of a monomial and an axiom, so has a $\NC_{r+1}$ derivation.
	
	Assume $X_{\vec{w}\vec{z}}\big(\sum_{\vec{a}\in V^t}X_{\vec{s}\vec{a}}-1\big)$ has a $\NC_{t+r}$ derivation. It can be easily verified that if a polynomial $f$ has a $\NC_r$ derivation from some set of axioms, then $xf$ has a $\NC_{r+1}$ derivation for any variable $x$. Thus, $x_{s'a'}X_{\vec{w}\vec{z}}\big(\sum_{\vec{a}\in V^t}X_{\vec{s}\vec{a}}-1\big)$ has a $\NC_{t+r+1}$ derivation. Finally
	$$
	\sum_{a'\in V}x_{s'a'}X_{\vec{w}\vec{z}}\big(\sum_{\vec{a}\in V^t}X_{\vec{s}\vec{a}}-1\big)=X_{\vec{w}\vec{z}}\big(\sum_{\vec{a}\in V^t, a'\in V}X_{(\vec{s}\cdot s')(\vec{a}\cdot a')}-1\big)=X_{\vec{w}\vec{z}}\big(\sum_{\vec{a}\in V^{t+1}}X_{(\vec{s}\cdot s')\vec{a}}-1\big)
	$$
	as required.
\end{proof}

\begin{lemma}\label{thirdone}
	There is a $\NC_k$ derivation over $\mathbb{Z}_p$ from $\mathrm{Ax}(\Gamma_{\vec{v}\rightarrow \vec{u}})$ of the polynomials in formulae (\ref{der1}) and (\ref{der2}).
\end{lemma}
\begin{proof}
	For $\vec{a}\in V^r$, set $\mathcal{N}(\vec{u},\vec{w})=\{\vec{a}\in V^r\mid \gamma(\vec{u}\langle \vec{i},\vec{w}\cdot\vec{a}\rangle)=\sigma\}$ and $\mathcal{N}(\vec{v},\vec{z})=\{\vec{a}\in V^r\mid\gamma(\vec{v}\langle \vec{i},\vec{a}\cdot\vec{z}\rangle)=\sigma\}$ (note the slight asymmetry). By Lemma \ref{ll1}, there is a $\NC_{2r}$ (and hence, $\NC_k$, since $2r<k$) derivation of
	$$
	X_{\vec{a}\vec{z}}\big(\sum_{\vec{a}'\in V^r}X_{\vec{s}\vec{a}'}-1\big)
	$$
	for every $\vec{a},\vec{s},\vec{z}\in V^r$. By subtracting from the above all monomials $X_{\vec{a}\vec{z}}X_{\vec{s}\vec{a}'}$ for which $\vec{a}'\notin \mathcal{N}(\vec{v},\vec{z})$ (which have a $\NC_k$ derivation by Lemma \ref{ll0}) one gets
$$
X_{\vec{a}\vec{z}}\big(\sum_{\vec{a}'\in \mathcal{N}(\vec{v},\vec{z})}X_{\vec{s}\vec{a}'}-1\big).
$$
Adding these for all $\vec{a}\in\mathcal{N}(\vec{u},\vec{w})$ yields
	\begin{equation}\label{derr1}
		\sum_{\vec{a}\in \mathcal{N}(\vec{u},\vec{w})}X_{\vec{a}\vec{z}}\big(\sum_{\vec{a}'\in \mathcal{N}(\vec{v},\vec{z})}X_{\vec{s}\vec{a}'}-1\big).
	\end{equation}
A similar argument shows that there is a $\NC_k$ derivation of
\begin{equation}\label{derr11}
	\sum_{\vec{a}\in \mathcal{N}(\vec{v},\vec{z})}X_{\vec{w}\vec{a}}\big(\sum_{\vec{a}'\in \mathcal{N}(\vec{u},\vec{w})}X_{\vec{a}'\vec{s}}-1\big).
\end{equation}
Subtracting (\ref{derr1}) from (\ref{derr11}) yields (\ref{der1}).

The polynomials in (\ref{der2}) can be derived by setting $r=0$ in Lemma \ref{ll1}.
\end{proof}
\begin{proof}[Proof of Theorem \ref{thm:main2}]
	If $\mathbb{F}=\mathbb{Q}$, and $X_{\vec{u}\vec{v}}$ has a $\PC_k$ derivation over $\mathbb{Q}$ from $\mathrm{Ax}(\Gamma)$, there is some $k'\in\mathbb{N}$ and a $\mathcal{C}_{k'}$ formula distinguishing $(\mathfrak{A}_\Gamma,\vec{z}\mapsto\vec{u})$ from $(\mathfrak{A}_\Gamma,\vec{z}\mapsto\vec{v})$ by Theorem 6.6 in \cite{g2p2}. By Corollary \ref{basiscase}, $\vec{u}$ and $\vec{v}$ are distinguished by $[\alpha_{k'+1,\Gamma}]^{\Sol_{k'+1}^0}$. It follows from Theorem \ref{fund} and Lemma \ref{secondone} that there is a $\NC_{k'+2}$ derivation over $\mathbb{Q}$ from $\mathrm{Ax}(\Gamma)$ of $X_{\vec{u}\vec{v}}$, as required by statement $(1)$.
	
	If $\mathbb{F}=\mathbb{Z}_p$ and $X_{\vec{u}\vec{v}}$ has a $\MC_k$ derivation over $\mathbb{Z}_p$ from $\mathrm{Ax}(\Gamma)$, there is some $k'$ and an $\mathrm{FOS}_{k'}(p)$ formula distinguishing $(\mathfrak{A}_\Gamma,\vec{z}\mapsto\vec{u})$ from $(\mathfrak{A}_\Gamma,\vec{z}\mapsto\vec{v})$ (Lemma \ref{4.1}). Theorem \ref{fund} and Lemma \ref{secondone} imply that there is a $\NC_{k'+1}$ derivation over $\mathbb{Z}_p$ from $\mathrm{Ax}(\Gamma)$ of $X_{\vec{u}\vec{v}}$, thus showing statement (2).
\end{proof}
\subsection{Generalized Cai-F\"urer-Immerman constructions}
In 1992, Cai, F\"urer and Immerman provided, for $k\geq 1$, a family of pairs of non-isomorphic graphs $(\mathcal{G}_k,\mathcal{H}_k)$ which cannot be distinguished by $\mathcal{C}_k$-formulae. Furthermore, the construction of such graphs ensures that $\mathcal{G}_k$ and $\mathcal{H}_k$ only have $O(k)$ vertices and require $\mathcal{C}$-formulae with $\Omega(k)$ variables in order to be distinguished, thus showing an optimal lower bound (up to constant factors) on the number of variables required to define such structures in the logic.

It is known, however, that $\mathcal{G}_k$ and $\mathcal{H}_k$ can be distinguished in poly($k$)-time. On the one hand, the construction ensures that both graphs, for any $k$, are $3$-regular, and isomorphism of graphs with bounded valency can be decided in polynomial time (see \cite{luks}). Alternatively, the isomorphism problem between $\mathcal{G}_k$ and $\mathcal{H}_k$ really encodes a system of linear equations over $\mathbb{Z}_2$, the solvability of which can be decided in polynomial time.

These structures can be generalized to encode a systems of linear equations over an arbitrary finite field (or, indeed, an arbitrary finite Abelian group \cite{holm}). Loosely speaking, the \emph{generalized Cai-F\"urer-Immerman} construction for the field $\mathbb{Z}_p$ provides, for each $k\in\mathbb{N}$, $p$ non-isomorphic $3$-regular graphs $\mathcal{G}_k^{(1)},\hdots,\mathcal{G}_k^{(p)}$. These delimit the power of well known linear algebra based polynomial-time approximations of graph isomorphism, as can be seen in the following known results, which we have paraphrased in terms of the equivalent problem of computing the orbits of the induced action of the automorphism group. Let $\Gamma_{k,p}$ be the disjoint union of the graphs $\mathcal{G}_k^{(1)},\hdots,\mathcal{G}_k^{(p)}$ and let $V$ denote its vertex set.
\begin{theorem}[Theorems 8.1 and 8.2 in \cite{dgp}]\label{thim}
	If $\Gamma=\Gamma_{k,p}$ and $q\neq p$, then the equivalence classes of $[\alpha_{k,\Gamma}]^{\IM_k^q}$ do not coincide with those of the $k$-orbit partition for $\Gamma$. If $q=p$, the equivalence classes of $[\alpha_{3,\Gamma}]^{\IM_3^q}$ coincide with those of the $3$-orbit partition for $\Gamma$.
\end{theorem}
The next statement can be deduced from the proof of Theorems 6.1 and 6.2 in \cite{g2p2}.
\begin{theorem}\label{ggpp}
	The equivalence classes of $\equiv_{\PC_k}^p$ on $V^k$ do not coincide with those of the $k$-orbit partition for $\Gamma_{k,q}$ if $q\neq p$.
\end{theorem}

Originally, the generalized Cai-F\"urer-Immerman constructions were introduced to delimit the expressive power of the extension of fixed point logics with rank operators over a finite field FPR$(p)$ and, correspondingly, the distinguishing power of the extension of first order logic by said operators, FOR$(p)$ (see Chapter 7 of \cite{holm}).  The distinguishing power of the invertible map test in characteristic $p$ is at least that of FOR$(p)$, so in one direction Theorem~\ref{thim} provides a lower bound for this logic.  In the other direction, we can show that the orbit partition on $\Gamma_{k,p}$ can already be defined in FPR$(p)$.  Indeed, it can be defined in the apparently weaker logic FPS$(p)$, giving the following result.

\begin{theorem}\label{needed}
	Let $\Gamma=\Gamma_{k,p}$. If $p\neq q$, there exist $\vec{u},\vec{v}\in V^k$ in different equivalence classes of the $k$-orbit partition of $\Gamma$, such that there are no FOS$_k(q)$-formulae distinguishing $(\mathfrak{A}_\Gamma,\vec{z}\mapsto\vec{u})$ from $(\mathfrak{A}_\Gamma,\vec{z}\mapsto\vec{v})$. If $q=p$ and $\vec{u},\vec{v}\in V^2$, then $(\mathfrak{A}_\Gamma,\vec{z}\mapsto\vec{u})$ is distinguished from $(\mathfrak{A}_\Gamma,\vec{z}\mapsto\vec{v})$ by some FOS$_2(q)$-formula if, and only if, $\vec{u},\vec{v}$ are in different equivalence classes of the $2$-orbit partition of $\Gamma$.
\end{theorem}

This allows us to show Theorem \ref{thm:main3}.
\begin{proof}[Proof of Theorem \ref{thm:main3}]
	Let $\Gamma=\Gamma_{k,p}$ and $\vec{u},\vec{v}$ be as in the statement of the Theorem. By Theorem \ref{needed} and Theorem \ref{fund}, $\vec{u}$ and $\vec{v}$ belong to different equivalence classes of $[\alpha_{3,\Gamma}]^{\Sol_3^p}$. By Lemma \ref{secondone} $\vec{u}\not\equiv_{\NC_3}^p\vec{v}$ as required.
\end{proof}
Theorem \ref{thm:main3} generalizes Theorem 6.3 in \cite{berkh} to aribitrary positive characteristic. Combined with Theorem \ref{ggpp}, it also shows that bounded degree $\PC,\MC$ and $\NC$ refutations have a similar distinguishing power as the invertible map tests on the class of generalized Cai-F\"urer-Immerman constructions.

\section{Conclusions: where does polynomial calculus lie?}

The invertible map tests can be thought of a family of algorithms, each of which distinguishes tuples of vertices of graphs according to the most general linear algebraic invariants expressible with a bounded number of variables in a logic. As bounded degree $\PC,\MC$ and $\NC$ refutations can be decided solely by using basic field operations, one expects that the equivalences defined by the invertible map tests simulate those defined by the abovementioned proof systems. Theorem \ref{thm:main1} gives a partial proof of this conjecture, leaving it open as to whether the invertible map tests can simulate bounded degree $\PC$ refutations, when taken over some finite field.

Our approach to this question was to attempt to define the above proof systems in the simplest logics which could express the solvability of systems of linear equations. The proof of Lemma \ref{4.1} hints at the flaws of this choice. Deciding whether or not there is a $\PC_k$ refutation of $P\subseteq \mathbb{F}[x_1,\hdots,x_r]$ can be understood as the following procedure, similar to that in Section 4.2.

\begin{enumerate}
	\item[INPUT.] $P\subset \mathbb{F}[x_1,\hdots,x_r]$
	\item[OUTPUT.] \texttt{REFUTE} or \texttt{NOREFUTE}.
	\item Initialize $\mathcal{S}=\{X_Af\mid f\in P, \mathrm{deg}(X_Af)\leq k\}$.
	\item \texttt{while} $\mathrm{span}_\mathbb{F}\mathcal{S}$ has changed since last round or $1\notin \mathrm{span}_\mathbb{F}\mathcal{S}$ \texttt{do lines 3-4}
	\item Find a set $\mathcal{B}$ generating the $\mathbb{F}$-space $\{f\in\mathrm{span}_\mathbb{F}\mathcal{S}\mid \mathrm{deg}(f)<k\}$.
	\item $\mathcal{S}\leftarrow \mathcal{S}\cup \{X_Af\mid f\in\mathcal{B}, \mathrm{deg}(X_Af)\leq k\}$.
	\item \texttt{If} $1\in \mathrm{span}_\mathbb{F}\mathcal{S}$ \texttt{output REFUTE}
	\item \texttt{else output NOREFUTE}.
\end{enumerate}

This procedure runs in polynomial time, as one can find $\mathcal{B}$ by using Gaussian elimination (there is no need to store the set $\mathrm{span}_\mathbb{F}\mathcal{S}$ as a generating set suffices), and the number of iterations of the \texttt{while} loop is bounded by the number of monomials of degree at most $k$ in the variables $\{x_1,\hdots,x_r\}$. If $\mathbb{F}$ is finite, this procedure is a priori \emph{not} definable in FPS$(p)$, as it is not immediate whether there is a canonical choice for the set $\mathcal{B}$ (its counterpart in the procedure for monomial calculus refutations was the set $\mathcal{M}$ of monomials in the span of $\mathcal{S}$). Put otherwise, defining the set $\mathcal{B}$ requires defining the solution space of a system of linear equations over a field of positive characteristic $p$, rather than just determining the solvability of the system and this cannot be done in FPS$(p)$.\footnote{One can use the extended Cai-F\"urer-Immerman constructions to explain this.} The FPC definability of bounded degree polynomial calculus refutations over the field $\mathbb{Q}$ (Theorem 4.9 in \cite{g2p2}) relies in fact on the FPC definability of solution spaces of linear equations over $\mathbb{Q}$ (Theorem 4.11 in \cite{g2p2}). It is important to note, however, that we are only concerned with axioms of the form $\mathrm{Ax}(\Gamma_{\vec{u}\rightarrow \vec{v}})$, and such a restriction could potentially allow a canonical choice of $\mathcal{B}$.

Let us view the problem from the viewpoint of proof complexity. It follows from Lemma \ref{secondone}, that if $\PC_k$ refutations over $\mathbb{Z}_p$ are definable in FPS$(p)$, then for every $k$, there is some $k'$ such that if $\mathrm{Ax}(\Gamma_{\vec{u}\rightarrow \vec{v}})$ has a $\PC_k$ refutation over $\mathbb{Z}_p$, then it has a $\NC_{k'}$ refutation over $\mathbb{Z}_p$. It is known that for all $n$ and for any field, there is a set of axioms on $n(n+1)$ variables which can be refuted by $\PC_3$ but require degree $o(n)$ to be refuted by $\NC$ (Theorem $6$ in \cite{buss}). Furthermore, this lower bound can be shown to be optimal. On the other hand, Buss et.~al.\ have shown that $\NC$ derivations can be used to simulate \emph{tree-like} $\PC$ derivations (see Theorems 5.3 and 5.4 in \cite{bussetal}) with only a small increase in degree. For a set of axioms of the form $\mathrm{Ax}(\Gamma_{\vec{u}\rightarrow \vec{v}})$, it is not known if any $\PC_k$ refutation of such can be converted into a tree-like refutation without incurring in an unbounded increase in degree. 

Finally, let us justify our conjecture; namely, that for any field, the equivalences on tuples of vertices of graphs defined by the invertible map tests simulate those defined by bounded degree $\PC$ refutations. First, Theorems \ref{thim}, \ref{ggpp} and \ref{thm:main3} show that the invertible map tests and bounded degree $\PC$ refutations have a similar distinguishing power on the class of generalized Cai-F\"urer-Immerman graphs, thus suggesting that there should be a tight connection between them. Moreover, the procedure we have described above for deciding $\PC_k$ refutations of a set of polynomials is \emph{linear algebraic invariant} in the sense that after each iteration of the \texttt{while} loop, the $\mathbb{F}$-span of $\mathcal{S}$ is independent of the choice of the set $\mathcal{B}$. This seems to be a strong hint at the fact that extending fixed-point logic with a quantifier over some linear algebraic operator over $\mathbb{Z}_p$ should suffice to define $\PC_k$ refutations over the same field.

\paragraph{Acknowledgements.} We are grateful to Martin Grohe, Benedikt Pago and Gregory Wilsenach for many helpful discussions and suggestions.

\bibliographystyle{alpha}
\bibliography{calculus}
\end{document}